\documentclass[11pt,english]{article}
\usepackage{amsmath,amssymb,amsthm}
\usepackage{multicol}
\usepackage[margin=1in]{geometry}
\usepackage{graphicx,color}
\usepackage{enumitem}
\usepackage{fullpage}
\usepackage[noblocks]{authblk}
\usepackage{tcolorbox}
\usepackage{babel}
\usepackage{wrapfig}
\usepackage{MnSymbol}
\usepackage{mdframed}
\usepackage{mathtools}
\usepackage{floatrow}
\usepackage[utf8]{inputenc}
\usepackage[ruled,linesnumbered,vlined]{algorithm2e}
\usepackage{tablefootnote}
\usepackage{thm-restate}
\usepackage{bbding}
\usepackage{pifont}
\usepackage{nicefrac}

\definecolor{Darkblue}{rgb}{0,0,0.4}
\definecolor{Brown}{cmyk}{0,0.61,1.,0.60}
\definecolor{Purple}{cmyk}{0.45,0.86,0,0}
\definecolor{Darkgreen}{rgb}{0.133,0.543,0.133}

\usepackage[colorlinks,linkcolor=Darkblue,filecolor=blue,citecolor=blue,urlcolor=Darkblue,pagebackref]{hyperref}
\usepackage[nameinlink]{cleveref}

\usepackage[colorinlistoftodos,prependcaption,textsize=tiny]{todonotes}

\newif\ifdraft 
\draftfalse

\newcommand{\namedref}[2]{\hyperref[#2]{#1~\ref*{#2}}}
\newcommand{\propref}[1]{\hyperref[#1]{property~(\ref*{#1})}}

\newtheorem{theorem}{Theorem}
\newtheorem{lemma}{Lemma}
\newtheorem{definition}{Definition}
\newtheorem{claim}{Claim}
\newtheorem{observation}{Observation}
\newtheorem{corollary}{Corollary}

\newtheorem{question}{Question}

\newcommand{\mst}{\mathrm{MST}}

\newcommand{\ddim}{\mathsf{ddim}}

\newcommand{\diam}{\mathrm{diam}}
\newcommand{\dm}{\mathrm{diam}}

\newcommand{\nd}{\delta}

\newcommand{\direct}{\overrightarrow}

\newcommand{\ball}{\mathbf{B}}

\newcommand{\eps}{\epsilon}

\newtheorem{fact}{Fact}
\definecolor{forestgreen}{rgb}{0.13, 0.55, 0.13}

\SetCommentSty{mycommfont}

\DeclareMathAlphabet{\mathpzc}{OT1}{pzc}{m}{it}
\newcommand{\etal}{{\em et al. \xspace}}

\newlength{\dhatheight}

\newcommand {\ignore} [1] {}

\newcommand{\initOneLiners}{%
	\setlength{\itemsep}{0pt}
	\setlength{\parsep }{0pt}
	\setlength{\topsep }{0pt}
}

\title{Greedy Spanners in Euclidean Spaces Admit Sublinear Separators}
\author{Hung Le}
\affil{University of Massachusetts at Amherst, \texttt{hungle@cs.umass.edu}}
\author{Cuong Than}
\affil{University of Massachusetts at Amherst, \texttt{cthan@umass.edu}}

\date{}
\begin{document}
\maketitle
\begin{abstract}
The greedy spanner in a low dimensional Euclidean space is a fundamental geometric construction that has been extensively studied over three decades as it possesses the two most basic properties of a good spanner: constant maximum degree and constant lightness.  Recently, Eppstein and Khodabandeh~\cite{EK21} showed that the greedy spanner in $\mathbb{R}^2$ admits a sublinear separator in a strong sense: any subgraph of $k$ vertices of the greedy spanner in $\mathbb{R}^2$ has a separator of size $O(\sqrt{k})$.  Their technique is inherently planar and is not extensible to higher dimensions. They left  showing the existence of a small separator for the greedy spanner  in $\mathbb{R}^d$ for any constant $d\geq 3$  as an open problem. 

In this paper, we resolve the problem of Eppstein and Khodabandeh~\cite{EK21} by showing that any subgraph of  $k$ vertices of the greedy spanner in $\mathbb{R}^d$ has a separator of size $O(k^{1-1/d})$. We introduce a new technique that gives a simple criterion for any geometric graph to have a sublinear separator that we dub \emph{$\tau$-lanky}: a geometric graph is  $\tau$-lanky if any ball of radius $r$ cuts at most $\tau$ edges of length at least $r$ in the graph. We show that any $\tau$-lanky geometric graph of $n$ vertices in $\mathbb{R}^d$ has a separator of size $O(\tau n^{1-1/d})$. We then derive our main result by showing that the greedy spanner is $O(1)$-lanky. We indeed obtain a more general result that applies to unit ball graphs and point sets of low fractal dimensions in $\mathbb{R}^d$.

Our technique naturally extends to doubling metrics. We use the $\tau$-lanky criterion to show that there exists a $(1+\eps)$-spanner for doubling metrics of dimension $d$ with a constant maximum degree and  a separator of size $O(n^{1-\frac{1}{d}})$;  this result resolves an open problem posed by Abam and Har-Peled~\cite{AH10} a decade ago. We then introduce another simple criterion for a graph in doubling metrics of dimension $d$ to have a sublinear separator. We use the new criterion to show that the greedy spanner of an $n$-point metric space of doubling dimension $d$  has a separator of  size $O((n^{1-\frac{1}{d}}) + \log\Delta)$ where $\Delta$ is the spread of the metric; the factor $\log(\Delta)$ is tightly connected to the fact that, unlike its Euclidean counterpart, the greedy spanner in doubling metrics has \emph{unbounded maximum degree}. Finally, we discuss algorithmic implications of our results.
\end{abstract}

\newpage

\setcounter{tocdepth}{2} 
\tableofcontents
    \newpage
    \pagenumbering{arabic}
    
\section{Introduction}

A \emph{$t$-spanner} of an edge-weighted graph $G= (V,E,w)$ is a spanning edge-weighted subgraph $H$ such that $\nd_{H}(u,v)\leq t \cdot \nd_{G}(u,v)$ for every pair of vertices $u,v \in V$; $t$ is called the \emph{stretch} of the spanner. This definition naturally extends to a $t$-spanner of a point set $P$ in a metric space $(X,\nd_X)$ where the graph $G = (P, {P \choose 2}, \delta_X)$ in this context is  the complete graph representing the metric induced by $P$. The study of spanners dated back to the work of Chew~\cite{Che86} in 1986 who constructed a $\sqrt{10}$-spanner for a point set on the Euclidean plane, though the term \emph{spanner} had not yet appeared until the work of Peleg and Sch\"{a}ffer~\cite{PS89}. 

Over more than 30 years, research on spanners has evolved into an independent field of study, and has found numerous applications, such as to VLSI circuit design \cite{CKRSW91,CKRSW92,SCRS01},  to distributed computing~\cite{Awerbuch85,ABP91,Peleg00}, to approximation algorithms~\cite{RS98,Klein05,Klein06,Gottlieb15,BLW17,CFKL20}, to wireless and sensor networks \cite{RW04,BDS04,SS10}, to machine learning \cite{GKK17}, and  to computational biology~\cite{RG05}.  Spanners enjoy wide applicability since they possess many desirable properties such as \emph{low sparsity}, implying, e.g., low storage cost; \emph{ small lightness}, implying, e.g., low construction cost; \emph{low maximum degree}, implying, e.g., small routing tables in routing applications. 
Here, the \emph{sparsity} of a spanner is the ratio of its number of edges to the number of edges of the minimum spanning tree ($\mst$), which is $n-1$, and the \emph{lightness} of a spanner is the ratio of its total edge weight to the weight of the $\mst$. Some spanners only have one property among three properties while others possess all of them and even more~\cite{ADMSS95,ES15}. 

The \emph{path greedy algorithm}, or greedy algorithm for short,  is perhaps the simplest and one of the most well-studied algorithms for constructing a $t$-spanner: consider edges of $G$ in the non-decreasing order of the weights, and add to the spanner an edge $e$ if the distance between its endpoints in the current spanner is larger than $t\cdot w(u,v)$. The output spanner is called the \emph{greedy spanner}. The greedy algorithm was introduced by Alth$\ddot{\mbox{o}}$fer \etal~\cite{ADDJS93} and independently discovered by Bern\footnote{Alth$\ddot{\mbox{o}}$fer \etal~\cite{ADDJS93} attributed the independent discovery of the greedy algorithm to an unpublished work of  Marshall Bern.}.  In addition to its great simplicity, the greedy algorithm has shown to be one-algorithm-fits-all. The greedy $(1+\eps)$-spanner has constant sparsity and constant lightness for point sets in Euclidean spaces\footnote{In this work, we assume that $\eps$ and $d$ are fixed constants, unless specified otherwise.}~\cite{CDNS92,DHN93,RS98,LS19}, doubling metrics~\cite{Smid09,Gottlieb15,BLW19}, and minor-free graphs~\cite{BLW19}. For \emph{edge-weighted} general graphs of $n$ vertices, the greedy spanner with stretch $(2k-1)$ has $O(n^{1/k})$ sparsity and with stretch  $(2k-1)(1+\eps)$ has $O(n^{1/k})$ lightness for a constant $\eps$, which are (nearly) optimal bounds assuming the Erd\H{o}s' girth conjecture~\cite{Erdos64}. Filtser and Solmon~\cite{FS16} showed that greedy spanners are \emph{existentially optimal} for both the size and the lightness for several graph families, which include general graphs and minor-free graphs. Roughly speaking,  existential optimality for a graph family $\mathcal G$ means that the worst-case sparsity and lightness of the greedy spanner over all graphs in $\mathcal G$  are as good as the sparsity and lightness of an optimal spanner over all graphs in $\mathcal G$. In summary,  understanding the greedy spanner is an important problem, and there have been many papers dedicated to this task.

In this work, we investigate greedy spanners for point sets  
from a different  perspective: 
the existence of sublinear balanced separators\footnote{We say that a separator $S$ of an $n$-vertex graph $G = (V,E)$ is balanced if every connected component of $G[V\setminus S]$ has at most $c\cdot n$ vertices for some constant $c$.}. For brevity, we simply refer to a  balanced separator as a separator.
While our work covers several settings, understanding the greedy spanner for point sets in Euclidean spaces is of special interest. This fundamental setting has been studied extensively; the chapter Geometric Analysis: The Leapfrog Property of the book by Narasimhan and Smid~\cite{NS07} is dedicated solely to  Euclidean spanners. Many algorithmic ideas developed in the context of Euclidean spanners could be carried over different settings, such as doubling metrics and unit ball graphs (UBGs), which are intersection graphs of unit balls in $\mathbb{R}^d$. On the other hand, having a small separator is algorithmically significant.  In their seminal work, Lipton and Tarjan~\cite{LT80} demonstrated the algorithmic power of sublinear separators in designing divide-and-conquer algorithms for solving many algorithmic tasks. (Though their results were stated for planar graphs, their techniques are applicable to any graph class that is closed under taking subgraphs and has sublinear separators.) Over four decades since the work of Lipton and Tarjan~\cite{LT80}, sublinear separators have played a central role in the development of many algorithmic paradigms, such as subexponential time parameterized algorithms~\cite{MP17,DFHT05,DH05,FLMPPS16}, analyzing local search heuristics~\cite{CG15,HQ17}, designing polynomial time approximation schemes (PTASes) for problems definable in first-order logic~\cite{Baker94,Dvo18}, to name a few. 

There have been some prior works that focus on (non-greedy) spanners with sublinear separators. Smith and Wormald~\cite{SW98} sketched a proof showing that the spanner of Arya \etal~\cite{ADMSS95} has a separator of size $O(n^{1-1/d})$; their proof relies on the \emph{$(\kappa,\lambda)$-thick property}\footnote{A set of objects (e.g., balls) in $\mathbb{R}^d$ is $(\kappa,\lambda)$-thick if every point in the space belongs to at most $\kappa$ objects and the maximum size (linear dimension, e.g., diameter) of an object is at most $\lambda$ times the minimum size.}. Abam and Har-Peled~\cite{AH10} constructed a $(1+\eps)$-spanner with a separator of size $O(n^{1-1/d})$ for point sets in $\mathbb{R}^d$ with maximum degree $O(\log^2 n)$ and constant sparsity using a semi-separated pair decomposition (SSPD), a concept introduced in the same paper. They left as an open problem of constructing a spanner with a sublinear separator and a \emph{constant maximum degree} in metrics of constant doubling dimensions. F{\"u}rer and Kasiviswanathan~\cite{FK12} constructed a spanner of constant sparsity with a separator of size $O(n^{1-1/d}) + \log(\Gamma))$ for ball graphs, which are intersection graphs of balls of arbitrary radii 
in $\mathbb{R}^d$; here $\Gamma$ is the ratio of the maximum radius to minimum radius over the balls, which could be exponentially large. However, for the important special case of UBGs ($\Gamma = 1$), the result of  F{\"u}rer and Kasiviswanathan~\cite{FK12} implies a separator of size $O(n^{1-1/d}))$. Sidiropoulos and Sridhar~\cite{SS17} devised a construction of a spanner with a separator of size $O(n^{1-1/d_{f}})$ for point sets in $\mathbb{R}^d$ with a \emph{fractal dimension} $d_f$ (see~\Cref{def:fractal} for a formal definition). Both results of F{\"u}rer and Kasiviswanathan~\cite{FK12} and Sidiropoulos and Sridhar~\cite{SS17} imply $(1+\eps)$-spanners with sublinear separators for point sets in $\mathbb{R}^d$ as their settings are more general\footnote{The fractal dimension of any point set in $\mathbb{R}^d$ is $d_f\leq d$.}. However, both constructions are fairly complicated and  based on  the $(\kappa,\lambda)$-thick property of Smith and Wormald~\cite{SW98}. Furthermore, their spanners are not known to have desirable properties of spanners such as maximum degree and constant lightness;  Abam and Har-Peled~\cite{AH10} showed that the max degree of spanner of F{\"u}rer and Kasiviswanathan~\cite{FK12} could be $\Omega(n)$. These results motivate the following question:

\begin{question}\label{ques:greedy} Does the greedy $(1+\eps)$-spanner for any point set in $\mathbb{R}^d$ have  a separator of size  $O(n^{1-1/d})$ for fixed constants $d\geq 2, \eps \leq 1$? More generally, do greedy $(1+\eps)$-spanners for more general settings, such as  ball graphs in $\mathbb{R}^d$ and point sets in Euclidean spaces of small fractal dimension, 
	have sublinear separators?
\end{question}

We remark that the upper bound $O(n^{1-1/d})$ of the separator size is optimal; the greedy $(1+\eps)$-spanner of point sets in the $d$-dimensional grid has size $\Omega(n^{1-1/d})$~\cite{SS17}.  

Recently, Eppstein and Khodabandeh~\cite{EK21} showed that the greedy spanner for point sets in $\mathbb{R}^2$ admits a separator of size $O(\sqrt{n})$, thereby resolving \Cref{ques:greedy} for the case of Euclidean graphs and $d = 2$. They derived their result by studying the crossing patterns of the straight-line drawing of the greedy spanner on the Euclidean plane. Since edge crossings are not useful in dimensions at least 3---the edges of the greedy spanners for point sets in a general position generally have no crossing---their technique cannot be applied to higher dimensions. They left \Cref{ques:greedy} for dimension $d\geq 3$ as an open problem. 

Nevertheless, the technique based on edge crossings of Eppstein and Khodabandeh~\cite{EK21} gives  a sublinear separator in a strong sense:  any subgraph of $k$ vertices of the greedy spanner in $\mathbb{R}^2$ has a separator of size $O(\sqrt{k})$. That is, having a square-root separator is a \emph{monotone property}: it is closed under taking subgraphs. Being a monotone property is not a mere generalization; all aforementioned algorithmic applications of sublinear separators require monotonicity. F{\"u}rer and Kasiviswanathan~\cite{FK12} showed that their spanner  has a monotone sublinear separator, and this is crucial for the algorithmic applications in their paper. It is not clear whether other non-greedy spanners mentioned above have the same property. On the other hand, all spanners in our paper have monotone sublinear separators. 

While resolving \Cref{ques:greedy} is an important goal, it is probably more important and challenging  to have a simple tool to resolve it, as asked in the following question:

\begin{question}\label{ques:characterization} Is there a simple criterion of graphs with sublinear separators in Euclidean and doubling metrics of constant dimensions that could be easily applied to (greedy or non-greedy) spanners?
\end{question}

A simple criterion asked in \Cref{ques:characterization} would be particularly significant, as  there are numerous spanner constructions in the literature that are faster than the greedy algorithm and give spanners with different properties than those guaranteed by greedy spanners. A prime example is greedy spanners for doubling metrics of constant dimensions. Smid~\cite{Smid09} constructed a doubling metric $(X,\nd_X)$ of dimension $1$ such that the greedy spanner for $(X,\nd_X)$ has maximum degree $n-1$. However, there are other constructions~\cite{CGMZ16,GL08} that give a spanner with constant maximum degree for doubling metrics. Does any of the spanners given by these algorithms have a sublinear separator? A positive answer to this question would resolve the open problem by Abam and Har-Peled~\cite{AH10}. And in this work, we provide a positive answer (\Cref{thm:CGMZDbouling}), by showing that the construction of~\cite{CGMZ16} gives a spanner with a sublinear separator (and constant maximum degree); we obtain our result based on a simple criterion that we develop. 

There have been several criteria of \emph{geometric graphs} with sublinear separators: $(\kappa,\lambda)$-thick property~\cite{SW98}, $k$-ply property of the intersection graphs of balls~\cite{MTTV97}, $k$-th nearest neighbor graphs~\cite{MTTV97}, intersection graphs of convex shapes that are tame~\cite{DMN21}; we refer readers to the cited papers for precise definitions of these concepts. While these criteria are fairly simple and generally applicable, it is not clear how to use them, except for the $(\kappa,\lambda)$-thick property~\cite{SW98}, in the context of spanners  where distances between points play a central role. (Sidiropoulos and Sridhar~\cite{SS17} and F{\"u}rer and Kasiviswanathan~\cite{FK12} used the $(\kappa,\lambda)$-thick property~\cite{SW98} in their constructions, but the algorithms become complicated.) Furthermore, except $k$-th nearest neighbor graphs~\cite{ZT09}, these criteria have been used for Euclidean spaces only and it is unclear how they could be useful to find sublinear separators of spanners in the more general setting of doubling metrics.

\subsection{Our Contributions}

Our first contribution is a  simple criterion of graphs with sublinear separators in Euclidean and doubling metrics, that we call \emph{$\tau$-lanky}. Roughly speaking, a graph is $\tau$-lanky if for every ball of radius $r$, there are at most $\tau$ edges of length at least $r$ is cut by the ball. Here, an edge is cut by a ball if exactly one endpoint of the edge is contained in the ball.

\begin{definition}[$\tau$-Lanky]\label{def:remotelyBoundedDeg} We say that a graph $G=(V,E,w)$ in a metric $(X,\delta_X)$ is \emph{$\tau$-lanky} if for any non-negative $r$, and for any ball $\ball_X(x,r)$ of radius $r$ centered at a vertex $x\in V$, there are at most $\tau$ edges of length at least $r$ that are cut by  $\ball_X(x,r)$. 
\end{definition}

We note that in \Cref{def:remotelyBoundedDeg},  we only count edges of length at least $r$. There could be $\Omega(n)$ edges of length less than $r$ that are cut by a ball of radius $r$.  We show in \Cref{subsec:boundeddeg} that lankiness implies the existence of a balanced sublinear separator. 

\begin{theorem}\label{thm:sublinearCharEucliden}Let $(X,\delta_X)$ be the Euclidean or a doubling metrics of constant dimension $d\geq 2$, and $G = (V,E,w)$ be an $n$-vertex graph in $(X,\delta_X)$ such that $G$ is $\tau$-lanky.  Then, $G$ has a balanced separator of size $O(\tau n^{1 - 1/d})$. Furthermore, the separator  can be found in $O(\tau \cdot n)$ expected time.
\end{theorem}

Since  $\tau$-lanky property is \emph{closed under taking subgraphs}, i.e., it is a monotone property,  \Cref{thm:sublinearCharEucliden} implies that any subgraph of $k$ vertices of a $\tau$-lanky graph has a separator of size $O(\tau k^{1-1/d})$. Thus, as a by-product, we obtain monotone sublinear separators. 

Unlike criteria introduced earlier for Euclidean spaces which are mostly based on intersections of objects such as balls or convex shapes, our criterion in \Cref{thm:sublinearCharEucliden} directly related to the edge length of $G$ via the notion of $\tau$-lanky. As a result, our criterion works better for spanners, where distances (and hence edge lengths) play a central role. Furthermore, as we will show later, for most spanners in this paper, it is relatively easy to see that they are $\tau$-lanky for some constant $\tau$ (that depends on the dimension and $\eps$). Thus, \Cref{thm:sublinearCharEucliden} can be seen as a positive answer to \Cref{ques:characterization}. 

Our second contribution is to show that greedy spanners of Euclidean graphs, unit ball graphs, and  point sets with small fractal dimensions are $\tau$-lanky. Our proof is relatively simple; see \Cref{sec:Eucliden} for details. Since $\tau$-lanky is monotone, \Cref{thm:sublinearCharEucliden} implies that:

\begin{theorem}\label{thm:EuclideanUBGsSfractal-Sep} Let $H$ be a subgraph of $k$ vertices of the  greedy $(1+\eps)$-spanner for a graph $G$. Then,
	\begin{enumerate}[nolistsep,noitemsep]
		\item[(1)] $H$ has a balanced separator of size $O(k^{1-1/d})$ if $G$ is a complete Euclidean graph representing a set of points in $\mathbb{R}^d$.
		\item[(2)] $H$ has a balanced separator of size $O(k^{1-1/d})$ if $G$ is a unit ball graph in $\mathbb{R}^d$.
		\item[(3)] $H$ has a balanced separator of size $O(k^{1-1/d_f})$ if $G$ is a complete graph  representing a set of points in $\mathbb{R}^d$ that has a fractal dimension $d_f\geq 2$.
	\end{enumerate}
Furthermore, the separator of $H$ can be found in $O(k)$ expected time given $H$. 
\end{theorem}

\Cref{thm:EuclideanUBGsSfractal-Sep} provides a positive answer to \Cref{ques:greedy}, and completely resolves the open problem raised by Eppstein and Khodabandeh~\cite{EK21}.   We also give a deterministic  linear time algorithm finding a sublinear separator  when the spread of the metric is polynomially large (\Cref{subsec:derandomize}). 

Our third contribution is to show that the classical $(1+\eps)$-spanner of Chan \etal~\cite{CGMZ16}  which is the first spanner for doubling metrics of constant maximum degree, has a sublinear separator. This result resolves the open problem raised by Abam and Har-Peled~\cite{AH10} a decade ago.  We obtain our result by showing that the spanner is $\tau$-lanky for a constant $\tau$. It then follows by \Cref{thm:sublinearCharEucliden} that:

\begin{theorem}\label{thm:CGMZDbouling} There exists a $(1+\eps)$-spanner in doubling metrics of  dimension $d\geq 2$  with a constant maximum degree such that any subgraph $H$ of $k$ vertices of the spanner has a balanced separator of size $O(k^{1-1/d})$. When $d = 1$,  $H$ has a balanced separator of size $O(\log k)$. \\
	Furthermore, the separator of $H$ can be found in $O(k)$ expected time given $H$. 
\end{theorem}

Finally, we study the greedy $(1+\epsilon)$-spanner for doubling metrics. By taking $r = 0$, the $\tau$-lanky property implies that the maximum degree of a $\tau$-lanky graph $G$ is at most $\tau$. However, Smid~\cite{Smid09} showed that the greedy spanner could have a maximum degree up to $n-1$, and thus, the lankiness parameter $\tau$ of the greedy spanner is $n-1$ in the worst case. On the other hand, the spread of the metric by Smid~\cite{Smid09} is $\Delta = 2^{\Omega(n)}$. Therefore, it is natural to ask: Does the greedy spanner for a doubling metric with a \emph{subexponential spread} have a sublinear separator? 
 We answer this question by introducing another property in our criterion: $\kappa$-thinness. A graph $G$ is $\kappa$-thin if for any  two sets of diameter at most $r$ whose distance is at least $r$ has $\kappa$ edges between them; see \Cref{def:kappa-separated-connect} for a formal definition. We then show (\Cref{thm:anotherSublinearSeparator}) that  any $\tau$-lanky and $\kappa$-thin graph in doubling metric of dimension $d$ has a separator of size $O(\kappa n^{1-1/d} + \tau)$. Indeed, our \Cref{thm:anotherSublinearSeparator} is stronger as it holds for a weaker notion of lankiness which does not impose the maximum degree of the graph by $\tau$. The size of the separator now depends on $\tau$ additively instead of multiplicatively as in \Cref{thm:sublinearCharEucliden}. By showing that the greedy spanner for doubling metrics is $\tau$-lanky and $\kappa$-thin for $\tau = O(\log \Delta)$ and $\kappa = O(1)$, we have that the greedy spanner of doubling metrics has a sublinear separator for any subexponential spread, as described in the following theorem.

\begin{restatable}{theorem}{StrongDoubilngSeparator}
	\label{thm:strong_doubling_separator}
	Let $(X, \nd_X)$ be a doubling metric of $n$ points with a constant doubling dimension $d$. Let $G$ be a $(1 + \eps)$-greedy spanner of $(X, \nd_X)$. $G$ has a balanced separator of size $O(n^{1-1/d} + \log(\Delta))$ when $d\geq 2$ and of size $O(\log(n) + \log(\Delta))$ when $d = 1$. 
\end{restatable}

\subsection{Algorithmic Implications}

One of the most popular applications of the spanners is to serve as overlay networks in  wireless networks~\cite{RW04,BDS04,SS10}. In these applications, we are often interested in solving computational problems in the spanner, such as shortest path~\cite{LWF03,GZ05}, independent set~\cite{Basagni2001,MM09}, dominating sets~\cite{MM09,PCA18}, connected dominating set~\cite{YWWY13}. The existence of sublinear separators in spanners implies that we can design provably good algorithms for these problems. 

In graph-theoretic terms, the spanners in \Cref{thm:EuclideanUBGsSfractal-Sep} and \Cref{thm:CGMZDbouling} have \emph{polynomial expansion}~\cite{DN16}. Har-Peled and Quanrud~\cite{HQ17} showed that many \emph{unweighted} optimization problems such as independent set, vertex cover, dominating set, connected dominating set, packing problems, admit a polynomial-time approximation scheme (PTAS) in graphs with polynomial expansion. Thus, all of these problems admit PTAS in our spanners. However, if the graphs have weights on vertices, the algorithm of Har-Peled and Quanrud~\cite{HQ17}, which is based on local search, does not have any guarantee on the approximation ratio. Indeed, designing a PTAS for vertex-weighted NP-hard problems in graphs with polynomial expansion remains an open problem~\cite{Dvo18}. We remark that planar graphs and minor-free graphs on which these problems were extensively studied~\cite{Baker94,Eppstein00,DH05,DHK05} are special cases of graphs with polynomial expansion. 

On the other hand, Dvo\v{r}\'{a}k~\cite{Dvo18} showed that if a graph has a polynomial expansion and \emph{bounded maximum degree}, then vertex-weighted optimization  problems considered above admit a PTAS. This implies that these problems admit a PTAS on the spanners considered in \Cref{thm:EuclideanUBGsSfractal-Sep} and \Cref{thm:CGMZDbouling} even when we have weights on the vertices, as these spanners have bounded maximum degree.  As approximating the vertex-weighted problems has been studied in many wireless network applications~\cite{SNS14,WB08,Basagni2001,Huang13}, our results could be of interest in these settings.

\section{Preliminaries}\label{sec:prelim}

We denote an edge-weighted graph by $G= (V,E,w)$ where $V$ is the vertex set, $E$ is the edge set, and $w: E(G)\rightarrow \mathbb{R}^+$ is the weight function on the edge set.  We denote a minimum spanning tree of $G$ by $\mst(G)$. When the graph is clear from context, we use $\mst$  as a shorthand for $\mst(G)$.  The {\em distance} between two vertices $p,q$ in $G$, denoted by $\nd_G(p,q)$, is the minimum weight of a path between them in $G$. The diameter of $G$,
denoted by $\dm(G)$, is the maximum pairwise distance in $G$. 

For a subset of vertices $A\subseteq V$, we denote by $G[A]$ the subgraph induced by $A$. We denote by $G-\{e\}$ a subgraph of $G$ obtained from $G$ by removing the edge $e$. For a subgraph $H$ of $G$, we define  $w(H) = \sum_{e\in E(H)}w(e)$ to be the total weight of edges in $H$. 	

A \emph{$c$-balanced separator} of a graph $G$ is a subset $S$ of $V$ such that every connected component of $G[V\setminus S]$ has at most $c|V|$ vertices. If $c$ is a constant, we simply refer to $S$ as a separator. 

In this paper, we study the spanner obtained by the (path) greedy algorithm. The following fact is well known for greedy spanners (see, e.g.,  Fact 6.1 in the full version of ~\cite{LS19})

\begin{fact}\label{fact:path-greedy}Let $H$ be a $t$-spanner of a graph $G = (V,E,w)$, then  $t\cdot w(x, y) \leq \nd_{H-e}(x,y)$ for any edge $e=(x,y)$ in $H$.
\end{fact}

We denote a metric space $X$ with a distance function $\nd_X$ by $(X, \nd_X)$. The diameter of a point set is the maximum distance between points in the set. A ball $\ball_X(p, r)$ centered at $p$ of radius $r$ is the set of points  within distance at most $r$ from the center $p$. We say that a ball $\ball_X(p, r)$ cuts an edge if exactly one endpoint of the edge is in $\ball_X(p, r)$. When the metric is clear from the context, we drop the subscript $X$ in the notation $\ball_X(p, r)$. 
We say that a set of points $P$ is \emph{$r$-separated} with $r > 0$ if the distance between any two points in $P$ is at least $r$.  

Given a set of points $P \in (X,\delta_X)$, we say that a subset $N\subseteq P$ is an \emph{$r$-net} of $P$ if $N$ is $r$-separated and for every point $y\in P$, there exists a point $x\in N$ such that $\delta_X(x,y)\leq r$. 

Given two sets of points $A$ and $B$, we denote by $\nd_X(A, B) = \min\limits_{a \in A, b \in B}{\nd_X(a, b)}$ the distance between $A$ and $B$.

\begin{definition}[$c$-Separated Pair]
	A pair of subsets $(A, B)$ in a metric $(X,\nd_X)$ is a \emph{$c$-separated pair} for some  $c > 0$  if the distance between $A$ and $B$ is at least $c$ times the maximum diameter of $A$ and $B$.
\end{definition}

If $(X,\nd_X)$ is the Euclidean metric, we use $|uv|$ to denote the distance between $u$ and $v$, and use $|A,B|$ to denote the distance between two point sets $A$ and $B$. Another metric studied in this paper is doubling metrics.

\begin{definition}[Doubling Metric]
	\label{def:doubling}
	A metric space $(X, \nd_X)$ has doubling constant $\lambda$ if any ball of radius $R$ can be covered by at most $\lambda$ balls of radius $\frac{R}{2}$. The number $\ddim = \log_2{\lambda}$ is called the doubling dimension of $(X, \nd_X)$.
\end{definition}

 It is well-known that the metric induced by any point set in $\mathbb{R}^d$ has doubling dimension $O(d)$. Doubling metrics and the Euclidean metric of dimension $d$ satisfy the following packing bound.

\begin{lemma}\label{lm:packing} Let $\ball(p,r)$ be a ball of radius $R$ in a Euclidean/doubling metric of dimension $d$, and $Y$ an $Y \subseteq \ball(p,x)$ be an $r$-separated subset for some $r\leq R$, then $|Y| = 2^{O(d)}\left(\frac{R}{r}\right)^d$.
\end{lemma}

We also consider the notion of \emph{fractal dimension} introduced by Sidiropoulos and Sridhard~\cite{SS17}. 

\begin{definition}[Fractal Dimension]
	\label{def:fractal}
	Let $P$ be a set of points in $\mathbb{R}^d$. $P$ has a \emph{fractal dimension} $d_f$ if and only if for any two positive numbers $r > 0, R \geq 2r$, any point $p \in \mathbb{R}^d$, and any $r$-net $N$ of $P$, $|N \cap \ball(p, R)| = O((R/r)^{d_f})$.  
\end{definition}

\section{Criteria of Graphs with Sublinear Separators}

In this section, we provide two criteria of graphs with sublinear separators in low dimensional Euclidean and doubling metrics: one for bounded degree graphs (\Cref{subsec:boundeddeg}), and one for graphs with high vertex degrees (\Cref{subsec:highdeg}).  Our proof uses the following lemma due to Har-Peled and Mendel~\cite{HPM06}.

\begin{lemma}[Lemma 2.4~\cite{HPM06}]\label{lm:smallBall} Let $P$ be a set of $n$ points in a metric $(X,\delta_X)$ with doubling constant $\lambda_X$. There exists a point $v\in P$ and a radius $r \geq 0$ such that (a)  $|\ball_X(v,r)\cap P| \geq \frac{n}{2\lambda_X}$ and  (b) $|\ball_X(v,2r)\cap P|\leq \frac{n}{2}$. Furthermore, $v$ and $r$ can be found in $O(\lambda_X^3 n)$ expected time. 
\end{lemma}

Before presenting these criteria in details, we introduce a notion of a \emph{packable metric space} (\Cref{def:packabeMetric}), which captures the packing bound (\Cref{lm:packing}) in both Euclidean and doubling metrics.  The notion of packable metric is very similar to the notion of fractal dimension in \Cref{def:fractal}; the main difference is that we do not restrict a packable metric  to being a submetric of an Euclidean metric. 

\begin{definition}[$(\eta,d)$-Packable Metric Space]\label{def:packabeMetric}
	A metric $(X,\delta_X)$ is \emph{$(\eta,d)$-packable} if for any $r\in (0,1]$ and any $r$-separated set $P\subseteq X$ contained in a \emph{unit ball}, $|P| \leq \eta\left(\frac{1}{r}\right)^d$.  \\
	We call $d$ the \emph{packing dimension} of the metric and $\eta$ the \emph{packing constant} of the metric.  
\end{definition}

A folklore result is that the doubling dimension of the Euclidean metric of dimension $d$ is $O(d)$. However, since the dimension of the metric will appear in the exponent of the separator, treating Euclidean metrics as a special case of doubling metrics would result in a \emph{polynomial loss} in the size of the separator. A strength of our technique is that, we only need the packing bound (\Cref{lm:packing}) in the construction of the separator. This property allows us to unify the construction of both doubling and Euclidean metrics via packable metrics.

By setting $R = 1$ in \Cref{lm:packing}, the Euclidean metric and doubling metrics of dimension $d$ are both $(\eta,d)$-packable for some $\eta = 2^{O(d)}$. We observe that the packing \Cref{lm:packing} also holds for $(\eta,d)$-packable metric space as well.

\begin{observation}\label{obs:packingPackable} Let $B$ be any ball of radius $R$ in a  $(\eta,d)$-packable metric space $(X,\delta_X)$, and $P\subseteq B$ be any $r$-separated set for some $r$ such that $0< r \leq R$. Then $|P| \leq \eta \left(\frac{R}{r}\right)^d$.
\end{observation}
\begin{proof} Scaling the metric by $\frac{1}{R}$, $B$ has radius $1$ and $P$ is $(r/R)$-separated. The observation now follows from \Cref{def:packabeMetric}.
\end{proof}

 We observe that the doubling dimension of a $(\eta,d)$-packable metric $(X,\nd_X)$ is also close to the packing dimension of $X$.

\begin{observation}\label{obs:doublingdimPack} Let $Y$ be a subset of points in an $(\eta,d)$-packable metric $(X,\delta_X)$. Let $\delta_Y$ be the restriction of the distance function $\delta$ on $Y\times Y$. Then $(Y,\delta_Y)$ has doubling constant $\lambda_Y \leq \eta 2^{d}$, and hence doubling dimension $\ddim_Y = d + \log(\eta)$.
\end{observation}
\begin{proof}Let $\ball_Y(v,R)$ be a ball of radius $R$ in $Y$ centered at a vertex $v\in Y$. Let $N\subseteq \ball_Y(v,R)$ be a $R/2$-net of   $\ball_Y(v,R)$. By \Cref{def:packabeMetric}, $|N|\leq \eta 2^{d}$. Since  $\ball_Y(v,R)$ can be covered by balls of radius $R/2$ centered at points in $N$, it follows that $\lambda_Y \leq \eta 2^{d}$.
\end{proof}

We say that an edge-weighted graph $G=(V,E,w)$ with $n$ vertices is \emph{a graph in a metric space} $(X,\delta_X)$ if $V\subseteq X$ and for any two vertices $u\not=v \in V$, $w(u,v) = \delta_X(u,v)$.  

\subsection{Bounded Degree Graphs}\label{subsec:boundeddeg}

We now introduce the \emph{lanky} property, and then we show that lanky graphs have sublinear separators. 

\begin{definition}[$\tau$-Lanky]\label{def:remotelyBoundedDegold} We say that a graph $G=(V,E,w)$ in a metric $(X,\delta_X)$ is \emph{$\tau$-lanky} if for any non-negative $r$, and for any ball $\ball_X(x,r)$ of radius $r$ centered at a vertex $x\in V$, there are at most $\tau$ edges of length at least $r$ that are cut by  $\ball_X(x,r)$. 
\end{definition}

Intuitively, if $G$ is $\tau$-lanky for a constant $\tau$, then for any ball $B$ of radius $r$, there is only a constant number of edges of length at least $r$ coming out from $B$.  We note that there could be as many as $\Omega(|E|)$ short edges of $G$ that are cut by $B$. We observe in the following that if $G$ is lanky, it has small degree.

\begin{observation}\label{obs:GDeg}  If $G$ $\tau$-lanky, then its maximum degree is at most $\tau$.
\end{observation}
\begin{proof}
	Let $v$ be any vertex in $G$, and let $r = 0$. By \Cref{def:remotelyBoundedDeg}, there are at most $\tau$ edges in $G$ of positive length that cut the ball $\ball(v,r)$, which only contains $v$. Thus, the degree $v$ is at most $\tau$. 
\end{proof}

We now show the main theorem in this section:  if $G$ is thin and lanky, it has a sublinear separator.

\begin{theorem}\label{thm:sublinearChar}Let $(X,\delta_X)$ is an $(\eta,d)$-packable metric space and $G = (V,E,w)$ is an $n$-vertex graph in $(X,\delta_X)$ such that $G$ is $\tau$-lanky.  Then, $G$ has a $\left(1-\frac{1}{\eta2^{d+1}}\right)$-balanced separator $S$ such that $|S| = O(\tau \eta 8^{d} n^{1 - 1/d})$ when $d\geq 2$ and $|S| = O(\tau \eta 8^{d}\log n)$ when $d= 1$. Furthermore, $S$  can be found in $O((\eta^38^d + \tau)n)$ expected time.
\end{theorem}

We first show the following lemma, which says that there exists a ball of radius $r^*$ that contains a constant fraction of vertices of $G$ and cut at most $O(n^{1-\frac{1}{d}})$ short edges of $G$. Our proof uses the random ball technique of Har-Peled~\cite{HarPeled11}; the same technique was used in  the construction of Sidiropoulos and Sridhar~\cite{SS17}.

\begin{lemma}
	\label{lm:shortEdgesCut} There exists a vertex $v\in V$ and a radius $r^*$ such that: 
		\begin{itemize}[noitemsep]
		\item[(1)] $\frac{n}{\eta 2^{d+1}}\leq |\ball_X(v,r^*) \cap V|\leq \frac{n}{2}$
		\item[(2)] $|E^*| = O(\tau \eta 8^{d} n^{1 - 1/d})$ when $d\geq 2$ and $|E^*| = O(\tau \eta 8^{d}\log n)$ when $d= 1$. Here $E^*$ is the set of all edges in $G$ of length at most $r^*$ that are cut by $\ball_X(v,r^*)$.
	\end{itemize}
Furthermore, $v$ and $r^*$ can be found in $O((\eta^3 8^{d}+ \tau)n)$ expected time.
\end{lemma}
\begin{proof}
Let $(V,\delta_V)$ be the submetric of $(X,\delta_X)$ induced by $V$. By \Cref{obs:doublingdimPack}, $(V,\nd_V)$ has doubling constant $\lambda_V = \eta 2^{d}$. By \Cref{lm:smallBall}, we can find in $O(\lambda_V^3 n) = O(\eta^38^d n)$ time a vertex $v$ and a radius $r$ such that: (a) $|\ball_V(v,r)| \geq \frac{n}{2\lambda_V} \geq \frac{n}{\eta 2^{d+1}}$ and (b) $|\ball_V(v,2r)| \leq \frac{n}{2}$.

Let $\sigma \in [0,1]$ be chosen uniformly at random. Let $r^* = (1+\sigma)r$; $ r^*\leq 2r$ since $\sigma \leq 1$. By properties (a) and (b), we have that $|\ball_X(v,r^*) \cap V| = |\ball_V(v,r^*)| \geq  |\ball_V(v,r)|\geq \frac{n}{\eta 2^{d+1}}$ and that $|\ball_X(v,r^*) \cap V| = |\ball_V(v,r^*)| \leq  |\ball_V(v,2r)|\leq \frac{n}{2}$. Thus, Item (1) is satisfied.

We now bound the expected size of $E^*$, the set of edges of length at most $r^*$ that are cut by $\ball_X(v,r^*)$. Let $E'$ be the set of edges of length at most $2r$ that are cut by $\ball_X(v, r^*)$. Then $E^*\subseteq E'$. We will bound the expected size of $E'$ instead, which implies the same bound on the expected size of $E^*$. We partition $E'$ into two sets $M_1,M_2$ where $M_1$ contains every edge of weight at most $rn^{-1/d}$ and $M_2 = E'\setminus M_1$. Observe that every edge $e$ of weight at most $rn^{-1/d}$ is cut by $\ball_X(v,r^*)$ with probability at most $\frac{rn^{1-d}}{r} = n^{-1/d}$. By \Cref{obs:GDeg}, $|E| \leq  \tau n/2$, and hence there are at most $\tau n/2$ edges of weight at most $rn^{-1/d}$.  Thus, it follows that:
\begin{equation}\label{eq:M1Expect}
	\mathbb{E}(|M_1|) \leq  n^{-1/d}(\tau n)/2 = \tau n^{1 - 1/d}/2
\end{equation}

We now bound the expected size of $M_2$.  For each $i \in [1,  \lceil\log(2n^{-1/d})\rceil]$, we define $r_i = 2^{i-1}n^{-1/d}r$ and a set of edges $M_2^i = \{(u, v) \in M_2|  r_i < w(u,v) \leq r_{i+1}\}$. Observe that $M_2 = \bigcup_{i=1}^{\lceil\log(2n^{-1/d})\rceil} M_2^i$. \hypertarget{ball_construction}{}Let $N_i$ be a $(r_i/2)$-net of $\ball_X(v,2r)$. By \Cref{obs:packingPackable}, we have that:
\begin{equation}\label{eq:NiSize}
	|N_i| \leq \eta\left(\frac{2r}{r_i/2}\right)^d =  \eta\left(\frac{4r}{2^{i-1}n^{-1/d}r}\right)^d = \eta\frac{8^dn}{2^{id}} 
\end{equation}
 Let $\mathcal{N}_i$ be set of balls with center in $N_i$ and radius $2^{i-2}n^{-1/d}r$. Since $r^*\leq 2r$, $\mathcal{N}_i$ covers all points in $\ball_X(v,r^*)$. Furthermore, since each ball in $N_i$ has diameter at most $r_i$, every edge of length at least $r_i$ (including edges in $M_i^{i}$) will be cut by at least one ball (and at most two balls) in $\mathcal{N_i}$. Note that the construction of $\mathcal{N}_i$ is deterministic. 
 
 Let $E_i$ be the set of edges of length at least $r_i$ and at most $2r_i$ that are cut by at least one ball in $\mathcal{N}_i$. Since $G$ is $\tau$-lanky, there are at most $\tau$ edges of length at least $r_i/2$ that are cut by a ball in $\mathcal{N}_i$. It follows that
 \begin{equation}\label{eq:MiSize}
 	|E_i| \leq \tau |N_i| = \tau \eta\frac{8^dn}{2^{id}} 
 \end{equation}
 by \Cref{eq:NiSize}. Since every edge in $e$ in $E_i$ is cut by $\ball_X(v,r^*)$ with probability at most $w(e)/r \leq r_{i+1}/r$. Thus, it follows that $\mathbb{E}(|M_2^i|) \leq |E_i| r_{i+1}/r$. By the linearity of expectation and \Cref{eq:MiSize}, we have:
 \begin{equation}
 	\begin{split}
 			\mathbb{E}(|M_2|) &=  \sum_{i = 1}^{\lceil\log(2n^{-1/d})\rceil} \mathbb{E}(|M_2^i|) = \sum_{i = 1}^{O(\log n)} \tau \eta\frac{8^dn}{2^{id}} \cdot \frac{r_{i+1}}{r}\\
 			& =  \tau \eta 8^d\sum_{i = 1}^{O(\log n)} \frac{n}{2^{id}}  \cdot \frac{2^{i}n^{-1/d}r}{r} =  \tau \eta 8^d n^{1 - 1/d}\sum_{i = 1}^{O(\log{n})}\left(\frac{1}{2^{d - 1}}\right)^i
 	\end{split}
\end{equation}

We conclude that $\mathbb{E}(|M_2|) = O(\tau \eta 8^{d}\log{n})$ when $d = 1$ and  $\mathbb{E}(|M_2|) = O(\tau \eta 8^{d}n^{1 - 1/d})$ when $d\geq 2$. By \Cref{eq:M1Expect}, we have that:
\begin{equation}\label{eq:EstarExpect}
\mathbb{E}[|E^*|] = \tau n^{1 - 1/d}/2 +  O(\tau \eta 8^{d}n^{1 - 1/d}) = O(\tau \eta 8^{d}n^{1 - 1/d}) 
\end{equation}
Thus, by Markov's inequality, with a constant probability, $E^*$ has size $O(\tau \eta 8^{d}n^{1 - 1/d})$.  Since the running time to find $E^*$ for each random choice of $r^*$ is $O(|E|) = O(\tau n)$, by repeating many times until we find $r^*$ such that $E^*$ satisfies Item (2), the expected running time is still $O(\tau n)$.  Recall that $v$ can be found in $O(\eta^38^d n)$  expected time. Thus, the total expected running time is  $O((\eta^38^d + \tau)n)$ as claimed.
\end{proof}

We are now ready to prove \Cref{thm:sublinearChar}.

\begin{proof}[Proof of \Cref{thm:sublinearChar}] Let $v$ and $r^*$  the center and the radius of the ball $\ball_X(v,r^*)$ as in \Cref{lm:shortEdgesCut}. We define the function $g(d,n)$ as follows $g(d,n) = O(\log(n))$ when $d = 1$ and $g(d,n) = O(n^{1-1/d})$ when $d\geq 2$.  Let $\tilde{E}$ be the set of edges that are cut by $\ball_X(v,r^*)$. By Item (2) in \Cref{lm:shortEdgesCut}, there are at most $O(\tau \eta 8^{d}g(d,n))$ edges in $\tilde{E}$ of length at most $r^*$. Since $G$ is $\tau$-lanky, there are at most $\tau$ edges of length at least $r^*$ in  $\tilde{E}$. Thus, it holds that
	\begin{equation} \label{eq:EtildeSize}
		|\tilde{E}| = O(\tau \eta 8^{d}g(d,n) + \tau)
	\end{equation}
Let $S$ be the set of all endpoints of edges in $\tilde{E}$. Then removing $S$ from $G$ disconnects the set of points in $\ball_X(v,r^*)$  from the set of points outside $\ball_X(v,r^*)$. Thus, $S$ is a $\left(1-\frac{1}{\eta 2^{d+1}}\right)$-balanced separator by Item (1) of \Cref{lm:shortEdgesCut}. The running time to construct $S$ is dominated by the running time to construct $v$ and $r^*$, which is  $O((\eta^38^d + \tau)n)$ by \Cref{lm:shortEdgesCut}.
\end{proof}

Observe by the definition that $\tau$-lanky property is closed under taking subgraph:  if $G$ is $\tau$-lanky, then any subgraph $H$ of $G$ is also $\tau$-lanky. Thus, by \Cref{thm:sublinearChar}, we have:

\begin{corollary}\label{cor:subgraphSepLanky} 
	Let $(X,\delta_X)$ is an $(\eta,d)$-packable metric space and $G = (V,E,w)$ is a graph in $(X,\delta_X)$ such that $G$ is $\tau$-lanky.  Let $H$ be any subgraph of $G$ with $k$ vertices.  Then, $H$ has a $\left(1-\frac{1}{\eta2^{d+1}}\right)$-balanced separator $S$ such that $|S| = O(\tau \eta 8^{d} k^{1 - 1/d})$ when $d\geq 2$ and $|S| = O(\tau \eta 8^{d}\log k)$ when $d= 1$. Furthermore, $S$  can be found in $O((\eta^38^d + \tau)k)$ expected time given $H$.
\end{corollary}

\subsection{Graphs with High Vertex Degrees}\label{subsec:highdeg}

By \Cref{obs:GDeg}, the $\tau$-lanky property implies that the maximum degree is bounded by $\tau$.
In this section, we introduce another criterion that could be used for graphs with high vertex degrees. Our criterion is based on two properties: weakly lanky and thin.

\begin{definition}[Weakly $\tau$-lanky]
	\label{def:weakly-tau-lanky}
	A graph $G = (V, E, w)$ in a metric $(X, \nd_X)$ is \emph{weakly $\tau$-lanky} if for any non-negative $r$, and for any ball $\ball_X(x,r)$ of radius $r$ centered at a vertex $x \in V$, there are at most $\tau$ vertices inside $\ball_X(x,r)$ that are incident to all edges of length at least $r$ cut by $\ball_X(x,r)$.
\end{definition}

Observe that $\tau$-lanky implies weakly $\tau$-lanky but the converse statement does not hold. There could be $\Omega(n)$ edges of length at least $r$ incident to a single vertex that is cut by a ball of radius $r$. We observe that the same proof in \Cref{subsec:boundeddeg} is applicable to yield a separator of size $O(\tau n^{1-1/d})$ for a weakly  $\tau$-lanky graph. However, in this section, we look for a separator of size $O(n^{1-1/d} + \tau)$. To this end, we introduce anther property that we call $\kappa$-thin.

\begin{definition}[$\kappa$-Thin]
	\label{def:kappa-separated-connect}
	A graph $G = (V, E, w)$ in a metric $(X, \nd_X)$ is $\kappa$-thin if for any $1$-separated pair $(A, B)$ of $V$, there are at most $\kappa$ edges between $A$ and $B$.
\end{definition}

Our goal in this section is to show the following theorem, which implies that if $G$ is $\kappa$-thin and weakly $\tau$-lanky, it has sublinear separator.  

\begin{theorem}
	\label{thm:anotherSublinearSeparator}
	Let $(X,\delta_X)$ is an $(\eta,d)$-packable metric space and $G = (V,E,w)$ is a graph in $(X,\delta_X)$ that is weakly $\tau$-lanky and $\kappa$-thin.  Then, $G$ has a $\left(1-\frac{1}{\eta2^{d+1}}\right)$-balanced separator $S$ such that $|S| = O(\eta^{O(1)}2^{O(d)}\kappa n^{1 - 1/d} + \tau)$ when $d\geq 2$ and $|S| = O(\eta^{O(1)}2^{O(d)}\kappa \log n + \tau)$ when $d= 1$. Furthermore, $S$ can be constructed in $O((\eta^38^d + \kappa)n)$ expected time.
\end{theorem}

In the proof of \Cref{thm:anotherSublinearSeparator}, we follow the same construction presented in \Cref{subsec:boundeddeg}: take a ball $\ball_X(v,r^*)$ of random radius $r^* \in [r,2r]$ centered at a specific vertex $v$, and construct the separator $S$ by taking the endpoint inside  $\ball_X(v,r^*)$ of every edge cut by the ball. We then show that in expectation, the size of $S$ is small.  There are three places in the proof of \Cref{thm:sublinearChar} where the $\tau$-lanky property is used to bound the size of the separator; here we point out how the $\kappa$-thin property could be used to replace $\tau$.  First, the number of edges in $E$ is bounded by $\tau n/2$, and this fact is used to bound the expected size of of edges $M_1$ in \Cref{eq:M1Expect}. We show in \Cref{lm:WSPDLinearEdge} below that  $\kappa$-thin property, implies that $|E| =  O(\kappa n)$. Second, the number of edges in $E$ of size at most $2r_i$  for some radius $r_i$ cut by some ball of radius at least $r_i$ is bounded by $\tau$, and this fact is used to bound the size of $E_i$ in \Cref{eq:MiSize}. We show in \Cref{lm:boundShortCutEdge} that the number of such edges is $O(\kappa)$ if $G$ is $\kappa$-thin, thereby, removing the depedency on $\tau$.  Finally, the $\tau$-lanky property is used to bound the number of edges of length at least $2r$ cut by $\ball_X(v,r^*)$, which contribute an additive $\tau$ in the size of $\tilde{E}$ in \Cref{eq:EtildeSize}. Since we can tolerate the additive term  $\tau$ in \Cref{thm:anotherSublinearSeparator}, we do not need to  do anything.

Now we focus on showing that $\kappa$-thin property implies that $G$ has $O(\kappa n)$ edges. Our proof uses well-separated pair decomposition. An  \emph{$s$-well-separated pair decomposition} ($s$-WSPD) for a point set $P$ for some non-negative parameter $s$ in a metric $(X,\nd_X)$ is the set of $s$-separated pairs $\mathcal{P} = \{(A_1,B_1),\ldots, (A_m,B_m)\}$ such that (i) $A_i, B_i \subseteq P, A_i \cap B_i = \emptyset$ for all $i \in [1, m]$ and (ii) for every pair of points $(p,q)$, there exists a some $i \in [1,m]$ such that either $p \in A_i, q\in  B_i$ or $q\in A_i$ and $p \in B_i$.

\begin{lemma}
	\label{lm:WSPDLinearEdge}
	If an $n$-vertex graph $G = (V, E, w)$ in an $(\eta,d)$-packable metric $(X, \nd_X)$ is $\kappa$-thin, then $|E| = O(\eta^{O(1)}2^{O(d)}\kappa n)$.
\end{lemma}
\begin{proof} Let $\mathcal{D} = \{(A_1, B_1), (A_2, B_2), \ldots (A_m, B_m)\}$ be a $1$-WSPD of $V$ with minimum number of pairs. Har-Peled and Mendel (Lemma 5.1~\cite{HPM06}) show that for any $n$-point set $P$ in a metric of doubling dimension $\ddim$ has an $s$-well-separated pair decomposition $\mathcal{P}$ with $|\mathcal{P}| = 2^{O(\ddim)}s^{\ddim}n$ for any $s\geq 1$. Since $X$ has doubling dimension $d + \log(\eta)$ by \Cref{obs:doublingdimPack}, $|\mathcal{D}| = \eta^{O(1)}2^{O(d)} n$. For each $i \in [1,m]$, let $H_{i}$ be the set of edges between $A_i$ and $B_i$. Since $(A_i, B_i)$ is $1$-separated, $|H_i| = \kappa$ by \Cref{def:kappa-separated-connect}. By the definition of well-separated pair decomposition, for each edge $e = (u, v) \in E$, there exists an index $i$ such that $(u, v)$ or $(v, u)$ is in $A_i \times B_i$. Thus, $\bigcup_{i = 1}^mH_i = E$. It follows that:
	\begin{equation}
		|E| = |\bigcup_{i = 1}^mH_i| \leq \sum_{i = 1}^m|H_i| = O(\eta^{O(1)}2^{O(d)}\kappa n),
	\end{equation}
	as claimed.
\end{proof}

Next, we show that, if $G$ is $\kappa$-thin, then the number of edges of length at least $r$ and at most $2r$ cut by a ball of radius $r$ is $O(\kappa)$. 

\begin{lemma}
	\label{lm:boundShortCutEdge} Let $\ball_X(x,r)$ be any ball of radius $r$ centered at some point $x \in X$. If $G$ is $k$-thin, there are are at most $O(\eta^22^{O(d)}\kappa)$ edges of $G$ of length at least $r$ and at most $2r$ that are cut by  $\ball_X(x,r)$. 
\end{lemma}

\begin{proof}
	Let $E_\text{short}$ be the set of edges of length in $[r, 2r]$. Observe by the triangle inequality that for every edge $(u, v) \in E_\text{short}$, both endpoints $u$ and $v$ are in $\ball_X(x, 3r)$. Let $\mathcal{B}_\text{in}$ ($\mathcal{B}_\text{out}$) be the set of balls obtained by taking balls of radius $\frac{r}{6}$ centered at points in a $\frac{r}{6}$-net of $\ball_X(x, r)$ ($\ball_X(x, 3r)$). By \Cref{obs:packingPackable}, $|\mathcal{B}_\text{in}| = O(\eta\left(\frac{r}{r/6}\right)^d) = O(\eta6^d)$ and $\mathcal{B}_\text{out} = O(\eta\left(\frac{3r}{r/6}\right)^d) = O(\eta18^d)$ that covers $\ball_X(x, 3r)$. Let $(B_\text{in}, B_\text{out}) \in \mathcal{B}_\text{in} \times \mathcal{B}_\text{out}$ be a pair of balls such that there exists an edge $(u, v) \in E_\text{short}$ between them. By the triangle inequality, $\nd_X(B_\text{in}, B_\text{out}) \geq \nd_X(u, v) - \diam(B_\text{in}) - \diam(B_\text{out}) \geq r - 2r/6 - 2r/6 = r/3$. Hence, $\frac{\nd_X(B_\text{in}, B_\text{out})}{\max\{\diam(B_\text{in}, B_\text{out})\}} \geq 1$, which implies that $(B_\text{in}, B_\text{out})$ is $1$-separated. Thus, there are at most $\kappa$ edges between $B_\text{in}$ and $B_\text{out}$ by the definition of $\kappa$-thin (\Cref{def:kappa-separated-connect}). It follows that $|E_\text{short}| = \kappa|\mathcal{B}_\text{in} \times \mathcal{B}_\text{out}| = O(\eta^2\kappa108^d) = O(\eta^2 2^{O(d)}\kappa)$, as claimed.
\end{proof}

Next, we show the following lemma, which is analogous to \Cref{lm:shortEdgesCut}. The key difference is that the size of $E^*$, the set of edges of length at most $r^*$ cut by  $\ball_X(v,r^*)$. 

\begin{lemma}
	\label{lm:anotherShortEdgesCut} Let $G = \{V, E, w\}$ is a weakly $\tau$-lanky and $\kappa$-thin graph in an $(\eta,d)$-packable metric space $(X,\delta_X)$. There exists a vertex $v\in V$ and a radius $r^*$ such that: 
	\begin{itemize}[noitemsep]
		\item[(1)] $\frac{n}{\eta 2^{d+1}}\leq |\ball_X(v,r^*) \cap V|\leq \frac{n}{2}$
		\item[(2)] $|E^*| = O(\eta^{O(1)}2^{O(d)} \kappa n^{1 - 1/d})$ when $d\geq 2$ and $|E^*| = O(\eta^{O(1)}2^{O(d)} \kappa\log n)$ when $d = 1$. Here $E^*$ is the set of all edges in $G$ of length at most $r^*$ that are cut by $\ball_X(v,r^*)$.
	\end{itemize}
	Furthermore, $v$ and $r^*$ can be found in $O(\eta^{O(1)}2^{O(d)} \kappa n)$ expected time.
\end{lemma}
\begin{proof}
	We reuse the notation in the proof of \Cref{lm:shortEdgesCut}. Specifically, we construct $v,r$ and $r^*$ as in \Cref{lm:shortEdgesCut}; Item (1) follows directly from the construction. Next, we bound the size of $E^*$ following the same strategy: partitioning $E^*$ into two set $M_1$ and $M_2$ where $M_1$ is the set of edges with length at most $rn^{1 - 1/d}$ and $M_2$ contains the other edges.  The same argument in  \Cref{lm:shortEdgesCut}, specifically \Cref{eq:M1Expect},  yields: 
	\begin{equation}\label{eq:M1ExpecThin}
		\mathbb{E}(|M_1|) \leq n^{-1/d}|E| \leq O(\eta^{O(1)}2^{O(d)}\kappa)n
	\end{equation}
	by \Cref{lm:WSPDLinearEdge}.
	
	To bound the expected size of $M_2$, we partition $M_2$ into $\{M_1,\ldots, M_{\lceil \log{(2n^{-1/d})} \rceil}\}$, where $M^i_2$ be the set of edges with length in $(r_i, r_{i + 1}]$; here  $r_i = 2^{i - 1}n^{1- 1/d}r$. Following the same \hyperlink{ball_construction}{construction} in \Cref{lm:shortEdgesCut}, we construct a set of $\mathcal{N}_i$ of balls of radius  $2^{i-2}n^{-1/d}r =r_i/2$ covering $\ball_X(v, 2r)$. By  \Cref{eq:NiSize},  $|\mathcal{N}_i| = \eta \frac{8^d n}{2^{id}}$. Let $E_i$ by the set of edges of length at least $r_i$ and at most $2r_i$ cut by at least one ball in $\mathcal{N}_i$. By \Cref{lm:boundShortCutEdge}, it follows that:
	\begin{equation}\label{eq:EiSizeThin}
		|E_i| \leq \tau |\mathcal{N}_i| = O(\eta^22^{O(d)}\kappa) \eta\frac{8^dn}{2^{id}} = O\left(\eta^3 2^{O(d)} \frac{n}{2^{id}}\right) 
	\end{equation}
	The rest of the argument is exactly the same as the proof of \Cref{lm:shortEdgesCut} (with different constants), that yields $\mathbb{E}[M_2] = O(\kappa \eta^{3}2^{O(d)}n^{1-1/d})$ when $d\geq 2$ and $\mathbb{E}[M_2] = O(\kappa \eta^{3}2^{O(d)}\log(n))$ when $d = 1$. This in turn implies $|E^*| = O(\eta^{O(1)}2^{O(d)} \kappa n^{1 - 1/d})$ when $d\geq 2$ and $|E^*| = O(\eta^{O(1)}2^{O(d)} \kappa\log n)$ when $d = 1$ as  desired.  The expected running time bound follow the same line of reasoning in the proof of \Cref{lm:shortEdgesCut}.
\end{proof}

We are now ready to complete the proof of \Cref{thm:anotherSublinearSeparator}. 

\begin{proof}[Proof of \Cref{thm:anotherSublinearSeparator}]
	Let $v$ and $r^*$ be the center and radius of the ball $\ball_X(v, r^*)$ in \Cref{lm:anotherShortEdgesCut}. Let $g(d, n) = O(\log{n})$ if $d = 1$ and $g(d, n) = O(n^{1 - 1/d})$ if $d \geq 2$. Let $S$ be the set of all endpoints in $\ball_X(v, r^*)$ of edges in $\tilde{E}$. Recall that $\tilde{E}$ is the set of edges that are cut by $\ball_X(v,r^*)$. Removing $S$ from $G$ disconnects $\ball_V(v, r^*)$ and $V \setminus \ball_V(v, r^*)$. By Item (2) of \Cref{lm:anotherShortEdgesCut}, there are $O(\eta^{O(1)}2^{O(d)}\kappa)g(d, n)$ edges of length at most $r^*$ cut by $\ball_X(v, r^*)$. By \Cref{def:weakly-tau-lanky}, there are at most $\tau$ points in $\ball_X(v, r^*)$ that are incident to all edges of length at least $r^*$ in $\tilde{E}$.  Thus, $|S| = O(\eta^{O(1)}2^{O(d)}\kappa)g(d, n) + \tau$. By Item (1) of \Cref{lm:anotherShortEdgesCut}, $S$ is a $\left(1 - \frac{1}{\eta 2^{d + 1}}\right)$-balanced separator of $G$. The expected time to find $v$ and $r^*$ is $O(\eta^{O(1)}2^{O(d)} \kappa n)$ by \Cref{lm:anotherShortEdgesCut}. Since the time to find all endpoints in $\ball_X(v, r^*)$ of edges in $\tilde{E}$ is $O(|E|) = O(\eta^{O(1)}2^{O(d)}\kappa n)$, the total time complexity to construct $S$ is $O(\eta^{O(1)}2^{O(d)}\kappa n)$.
\end{proof}

Since  weakly $\tau$-lanky and $\kappa$-thin properties are closed under taking subgraphs, we have:

\begin{corollary}
	\label{cor:anotherSubGSeparator}
	Let $(X,\delta_X)$ is an $(\eta,d)$-packable metric space and $G = (V,E,w)$ is a graph in $(X,\delta_X)$ such that $G$ is weakly $\tau$-lanky and $\kappa$-thin.  Let $H$ be any subgraph of $G$ with $k$ vertices.  Then, $H$ has a $\left(1-\frac{1}{\eta2^{d+1}}\right)$-balanced separator $S$ such that $|S| = O(\eta^{O(1)}2^{O(d)}\kappa k^{1 - 1/d} + \tau)$ when $d\geq 2$ and $|S| = O(\eta^{O(1)}2^{O(d)}\kappa\log k + \tau)$ when $d= 1$. Furthermore, $S$  can be found in $O(\eta^{O(1)}2^{O(d)}\kappa k)$ expected time given $H$.
\end{corollary}

\section{Separators of Greedy Spanners for Graphs in Euclidean Spaces}\label{sec:Eucliden}

In this section, we prove \Cref{thm:EuclideanUBGsSfractal-Sep}. Specifically, in \Cref{subsec:Euclidean}, we focus on the greedy spanners of point sets in Euclidean spaces. In \Cref{subsec:fractal}, we focus on the greedy spanners of point sets that have a low fractal dimension. In \Cref{subsec:UBG}, we focus on the greedy spanners of unit ball graphs (UBGs).

\subsection{Euclidean Spaces}\label{subsec:Euclidean}

First, we focus on proving \Cref{thm:strong_euclidean_separator} below, which implies  Item (1) of \Cref{thm:EuclideanUBGsSfractal-Sep}.

\begin{restatable}[Separators for Subgraphs of Greedy Spanners]{theorem}{StrongEuclideanSeparator}
	\label{thm:strong_euclidean_separator}
	Let $P$ be a given set of points in the $d$-dimensional Euclidean space and $G$ be the greedy $(1+\epsilon)$-spanner of $P$ for some $\eps \in (0,1/2]$. Then any $k$-vertex subgraph $H$ of $G$ has a $\left(1 - \frac{1}{2^{O(d)}}\right)$-balanced separator $S$ of size $O(2^{O(d)}\eps^{1 - 2d}k^{1-1/d})$. Furthermore, $S$ can be found in expected $O((2^{O(d)} + \eps^{1 - 2d})k)$ time given $H$.
\end{restatable}

Our goal in the proof of \Cref{thm:strong_euclidean_separator} is to show that the greedy $(1+\epsilon)$-spanner $G$ is $\tau$-lanky for some $\tau$ that depends on $\eps$ and $d$ only.  To this end, we rely on the following  lemma, which bounds the number of edges in the greedy $(1+\epsilon)$-spanner between subsets of points  $X$ and  $Y$, where $X$ has a small diameter and $Y$ is sufficiently far from $X$.

\begin{lemma}
	\label{lm:XtoYEdges}
	Let $P$ be a set of $n$ points in $\mathbb{R}^d$ for a constant $d$ and $G$ be a greedy $(1+\epsilon)$-spanner of $P$ for $\eps \in (0,1/2]$. Let $X$ and $Y$ be two subsets of $P$ such that $\diam(X) \leq \frac{\eps R}{12}$ and $|X,Y| \geq R$. Then $G$ has $O(\epsilon^{1-d})$ edges between $X$ and $Y$.
\end{lemma}

Our proof of \Cref{lm:XtoYEdges} uses the following lemma.
\begin{lemma}
	\label{lm:dist_in_small_cone}
	Given a point $x$ in $\mathbb{R}^d$ and an arbitrarily small number $\eps \in (0, 1/2]$. Let $C$ be a cone with apex $x$ and angle at most $\eps / 8$. Then, for any two points $y$ and $z$ in $C$ such that $|xy| \geq |xq|$, $|xy| \geq |yq|  + (1 - \eps / 4)|xq|$.
\end{lemma}

\begin{proof} Consider the triangle $xyq$ and let $h$ be the projection of $q$ on $xy$. Since $|xy| \geq |xq|$, $h$ is in the segment between $x$ and $y$.  Let $\gamma = \angle yxq$; see \Cref{fig:cone}(b). Observe that $\gamma \leq \theta = \epsilon/8 < 1$ when $\eps\in (0,1/2)$. We have:
	\begin{equation*}
		\begin{split}
			|xy| - |yq| &\geq |xy| - |yh| - |hq| = |xh| - |hq| = (\cos{\gamma} - \sin{\gamma})|xq|\\
			&\geq   (\cos{\theta} - \sin{\theta})|xq| \qquad \mbox{($\cos(x) - \sin(x)$ is monotonically decreasing for $x\in (0,1)$)}\\
			&\geq  (1 - 2\theta)|xq| = (1-\eps/4)|xq|~.
		\end{split}    
	\end{equation*}
	The penultimate  inequality is due to that $\cos(x) - \sin(x) \geq 1 - 2x$ every $x\in [0,1]$. Thus, we conclude that $|xy| \geq |yq| + (1-\eps/4)|xq|$ as claimed.	
\end{proof}	

We are now ready to prove \Cref{lm:XtoYEdges}.

\begin{proof}[Proof of \Cref{lm:XtoYEdges}]
	Let $x$ be an (arbitrary) point in $X$. We cover the space  $\mathbb{R}^d$ by a minimum number of cones with apex $x$, each with an angle $\theta = \epsilon/8$;  the number of cones in the cover is $O(d^{(d + 1)/2}(\pi / \theta)^{d - 1}) = O(\epsilon^{1-d})$ \cite{NS07}.  We show below that for each cone, there is at most \emph{one edge} in $G$ between $X$ and the subset of points of $Y$ in the cone. It then follows that the number of edges in $G$ between $X$ and $Y$ is $O(\epsilon^{1-d})$.
	
	\begin{figure}[htb]
		\center{\includegraphics[width=0.9\textwidth]{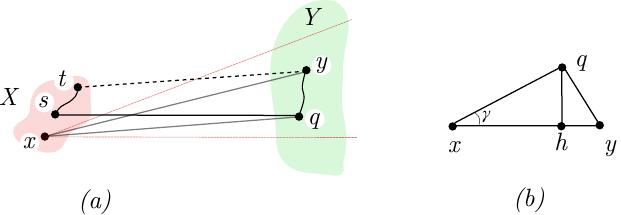}}
		\caption{(a) Points $y$ and $q$ of $Y$ are in the same cone bounded by two red dashed lines at apex $x$ and (b) the triangle $xqy$ with $\angle yxq = \gamma$.}
		\label{fig:cone}
	\end{figure}

	Suppose for a contradiction that there two edges $(s,q)$ and $(t,y)$ in $G$, where $s,t \in X$ and $q,y \in Y$. W.l.o.g, we assume $|xy| \geq |xq|$; see \Cref{fig:cone}(a). Our goal is to show that:
	\begin{equation}\label{eq:ty-contradiction}
		(1 + \eps)|ty| >  \nd_{G-(t,y)}(t,y)
	\end{equation} 
	which contradicts \Cref{fact:path-greedy}. From \Cref{lm:dist_in_small_cone}, we obtain $|xy| \geq |yq| + (1 - \eps / 4)|xq|$. Observe that:
	\begin{equation*}
		\begin{split} (1+\eps) |ty| &\geq (1+\eps)(|xy| - |tx|) \qquad \mbox{(by the triangle inequality)}\\
			& \geq (1+\eps)(|yq| + (1-\eps/4)|xq|) - (1+\eps) |tx| \qquad \mbox{(by \Cref{lm:dist_in_small_cone})}\\
			& > (1+\eps)|yq| + (1+\eps/2)|xq| - (1+\eps)|tx|  \qquad \mbox{(since $0 < \eps < 1$)}\\
			&\geq (1+\eps)|yq| + (1+\eps/2)(|sq| - |xs|)  - (1+\eps)|tx|  \qquad \mbox{(by the triangle inequality)}\\
			&\geq (1+\eps)|yq| + |sq| + \eps R/2 - (1+\eps/2) |xs|   - (1+\eps)|tx| \qquad \mbox{(since $|sq| \geq |X,Y| \geq R$)}\\
			&\geq (1+\eps)|yq| + |sq| + (1+\eps)|ts| + \eps R/2 - \underbrace{((1+\eps)|ts| +  (1+\eps/2) |xs|   + (1+\eps)|tx|)}_{\leq~ 3(1+\eps) \eps R/12 ~=~ \eps(1+\eps) R/4 \text{ since }\dm(X)\leq ~\eps R/12}\\
			&\geq  (1+\eps)|yq| + |sq| + (1+\eps)|ts| + \eps R/2 - (1+\eps)\eps R/4 \\
			&>  (1+\eps)|yq| + |sq| + (1+\eps)|ts| \qquad \mbox{ (since $\eps \leq 1/2$)}\\
		\end{split}
	\end{equation*}
	Since $G$ is a $(1+\eps)$-spanner, $(1+\eps)|yq| \geq \nd_G(y,q)$ and $(1+\eps) |ts| \geq \nd_G(t,s)$. Thus, we have:
	\begin{equation*}
		(1+\eps) |ty| > \nd_G(y,q) + \nd_G(s,q)  + \nd_G(t,s)
	\end{equation*}  
	Observe that $\max\{|yq|, |st|\} \leq \dm(X) \leq \eps R/12 ~\leq~ \eps |X,Y|/12 \leq \eps |ty|/12$. This implies $(1+\eps)\max\{|yq|, |st|\} < |ty|$ when $\eps \in (0,1/2]$; that is, edge $(t,y)$ cannot be in the shortest paths from $y$ to $q$ and from $s$ to $t$ in $G$. Thus, $\nd_G(y,q) = \nd_{G-(t,y)}(y,q)$ and $\nd_G(t,s) = \nd_{G-(t,y)}(t,s)$. Furthermore, since $(s,q)$ is an edge in $G$, $\nd_G(s,q) = \nd_{G-(t,y)}(s,q)$.  Thus, it holds that:
	
	\begin{equation*}
		(1+\eps) |ty|  >  \nd_{G-(t,y)}(y,q) + \nd_{G-(t,y)}(s,q)  + \nd_{G-(t,y)}(t,s) = \nd_{G-(t,y)}(t,y)~.
	\end{equation*}
	That is, \Cref{eq:ty-contradiction} holds as claimed. 
\end{proof}

\begin{lemma}\label{lm:EuiclideanGreedyLanky}
	Let $P$ be a set of $n$ points in $\mathbb{R}^d$ for a constant $d$ and $G$ be a greedy $(1+\epsilon)$-spanner of $P$ for $\eps \in (0,1/2]$. Then $G$ is $O(\eps^{1-2d})$-lanky.  
\end{lemma}
\begin{proof}
	Let $\ball(x,r)$ be a ball of radius $r$ centered at a vertex $x\in P$. Let $\tilde{E}$ be the set of edges of length at least $r$ that are cut by $\ball(x,r)$.   Let $\mathcal{B}$ be a minimum collection of balls of radius $\frac{\eps r}{48}$ that covers $\ball(x,r)$. By \Cref{lm:packing}, $|\mathcal{B}| = O(\eps^{-d})$.   We observe that, since every edge in $\tilde{E}$ has length at least the diameter of every ball in $\mathcal{B}$, each edge in $\tilde{E}$ must be cut by some ball in $\mathcal{B}$. 
	
	Consider a ball $B \in \mathcal{B}$, and $e = (u,v)\in \tilde{E}$ be an edge that is cut by $B$. W.l.o.g., we assume that $u\in B$. Then by the triangle inequality, $|v, B|\geq w(e) - \diam(B) \geq r - \frac{\eps r}{24}\geq r/2$ when $\eps \leq 1/2$. Since $\diam(B) = \frac{r\eps}{24}$, by \Cref{lm:XtoYEdges}, there are at most $O(\eps^{1-d})$ such edges $e$ in $\tilde{E}$ that cut $B$.  Thus, $|\tilde{E}|\leq |\mathcal{B}|O(\eps^{1-d}) = O(\eps^{1-2d})$; the lemma follows.
\end{proof}

\begin{proof}[Proof of \Cref{thm:strong_euclidean_separator}] When $d = 1$, the theorem holds trivially. Thus, we assume that $d\geq 2$. By \Cref{lm:EuiclideanGreedyLanky}, $G$ is $\tau$-lanky with $\tau = O(\eps^{1-2d})$. Since the Euclidean space of dimension $d$ is $(\eta, d)$-packable with $\eta = 2^{O(d)}$ by \Cref{lm:packing}, by \Cref{cor:subgraphSepLanky}, $H$ has a separator $S$ of size $O(\eps^{1-2d}2^{O(d)}k^{1-1/d})$. Also by \Cref{cor:subgraphSepLanky}, $S$ can be found in $O((2^{O(d)}8^d + \eps^{1-2d})k)$ time.
\end{proof}

\subsection{Point Sets of Low Fractal Dimensions}\label{subsec:fractal}

We now consider a point set $P$ in $\mathbb{R}^d$ that has a small fractal dimension. Since $P$ is still a set of points in  $\mathbb{R}^d$, \Cref{lm:EuiclideanGreedyLanky} remains true for the greedy $(1+\eps)$-spanner $G$ of $P$. The only difference is that, the metric induced by $P$, by \Cref{def:fractal}, is an $(O(1), d_f)$-packable metric. This implies a separator of smaller size. 

\begin{restatable}{theorem}{StrongFractalSeparator}
	\label{thm:strong_factal_separator}
	Let $P$ be a given set of points in $\mathbb{R}^d$ that has fractal dimension $d_f$, and $G$ be the greedy $(1+\epsilon)$-spanner of $P$. Then any $k$-vertex subgraph $H$ of $G$ has a $\left(1 - \frac{1}{2^{O(d_f)}}\right)$-balanced separator $S$ of size $O(2^{O(d_f)}\eps^{1 - 2d}k^{1-1/d_f})$. Furthermore, $S$ can be found in expected $O((2^{O(d_f)} + \eps^{1 - 2d})k)$ time given $H$.
\end{restatable}
\begin{proof}
	By \Cref{def:fractal} and \Cref{def:packabeMetric}, the metric induced by $P$ is $(O(1), d_f)$-packable. By \Cref{lm:EuiclideanGreedyLanky}, $G$ is $O(\eps^{1 - 2d})$-lanky. By \Cref{cor:subgraphSepLanky}, every $k$-vertex subgraph $H$ of $G$ has a $\left(1 - \frac{1}{2^{O(d_f)}}\right)$-balanced separator of size $O(2^{O(d_f)}\eps^{1 - 2d}k^{1-1/d_f})$ that can be found in expected $O((2^{O(d_f)} + \eps^{1 - 2d})k)$ time.  	
\end{proof}

We note that \Cref{thm:strong_factal_separator} implies Item (3) in \Cref{thm:EuclideanUBGsSfractal-Sep}. 

\subsection{Unit Ball Graphs}\label{subsec:UBG}

In this section, we show that the greedy spanners of UBGs have sublinear separators, as described in the following theorem. We note that \Cref{thm:strong_unibal_separator} implies Item (2) of \Cref{thm:EuclideanUBGsSfractal-Sep}. 

\begin{restatable}{theorem}{StrongBallSeparator}
	\label{thm:strong_unibal_separator}
	Let $G$ be the greedy $(1+\epsilon)$-spanner of a unit ball graph in $\mathbb{R}^d$. Then any $k$-vertex subgraph $H$ of $G$ has a $\left(1 - \frac{1}{2^{O(d)}}\right)$-balanced separator $S$ of size $O(2^{O(d)}\eps^{1 - 2d}k^{1-1/d})$. Furthermore, $S$ can be found in expected $O((2^{O(d)} + \eps^{1 - 2d})k)$ time given $H$.
\end{restatable}

Unit ball graphs are intersection graphs of unit balls. Since we scale the metric so that the minimum inter-point distance is $1$, we assume that the unit ball graphs are the intersection graphs of balls of radius $\mu$ for some positive $\mu$.

We consider the greedy spanners of unit ball graphs. We prove that a greedy spanner of a unit ball graph in a $d$-dimensional Euclidean space also admits $\tau$-lanky property. We prove a lemma which is analogous to \Cref{lm:XtoYEdges}.

\begin{lemma}
	Let $P$ be a set of $n$ points in $\mathbb{R}^d$ for a constant $d$. Let $G_B$ be the unit ball graphs with the set of centers $P$ and $G$ be a greedy $(1+\epsilon)$-spanner of $G_B$ for $\eps \in (0,1/2]$. Let $X$ and $Y$ be two subsets of $P$ such that $\diam(X) \leq \frac{\eps R}{12}$ and $|X,Y| \geq R$. Then $G$ has $O(\epsilon^{1-d})$ edges between $X$ and $Y$.
\end{lemma}

\begin{proof}
	 Let $x$ be an arbitrary point in $X$. We partition the space $\mathbb{R}^d$ into $O(\eps^{1 - d})$ cones with apex $x$ and angle $\theta = \eps/8$. We prove that for each cone $C$, there is at most one edge in $G$ from $X$ to $Y \cap C$. We follow the same proof strategy of \Cref{lm:XtoYEdges}. Assume that there are two edges $(s, q)$ and $(t, y)$ in $G$ with $s, t \in X$ and $q, y \in Y$ (see \Cref{fig:cone}). W.o.l.g., we suppose that $|xy| \geq |xq|$. Our goal is to show that $(1 + \eps)d\delta_{G_B}(t, y) > \delta_{G - (t, y)}(t, y)$, which contradicts \Cref{fact:path-greedy}.  
	 
	 First, we prove that $(q, y) \in E(G_B)$. By \Cref{lm:dist_in_small_cone}, we obtain $|xy| - |yq| \geq (1 - \eps / 4)|xq|.$ Using the triangle inequality, we have that:
	 \begin{equation}
	 	\begin{split}
	 		|yq| &\leq \underbrace{|xy|}_{\leq |ty| + |xt|} - (1 - \eps/4)\underbrace{|xq|}_{\geq |sq| - |sx|} \leq |ty| + \underbrace{|xt|}_{\leq \frac{\eps R}{12}} + (1 - \eps/4)\underbrace{|sx|}_{\leq \frac{\eps R}{12}} - (1 - \eps/4)\underbrace{|sq|}_{\geq R}\\
	 		&\leq |ty| + \frac{\eps R}{12} + (1 - \eps/4)\frac{\eps R}{12} - (1 - \eps/4)R \leq |ty| + \frac{\eps R}{12} + \frac{\eps R}{12} - \frac{3R}{4} \leq |ty| \leq 2\mu \quad \mbox{(since $\eps \leq 1/2$)}.
	 	\end{split}
	 \end{equation}
	The bound $|yq| \leq 2\mu$	implies that $y$ is adjacent to $q$ in $G_B$. For any pair of vertices $(v_1, v_2)$ with $v_1, v_2 \in X, v_1 \neq v_2$, we observe that $|v_1v_2| \leq \frac{\eps R}{12} \leq R \leq |ty| \leq 2\mu$. Hence, $v_1$ is adjacent to $v_2$ in $G_B$. Then, we have $(t, s) \in E(G_B)$. In \Cref{lm:XtoYEdges}, we have prove that $(1 + \eps)|ty| > (1 + \eps)|yq| + |sq| + (1 + \eps)|ts|$. Since $G$ is a $(1 + \eps)$-spanner of $G_B$, $\nd_G(y, q) \leq (1 + \eps)\nd_{G_B}(y, q) = (1 + \eps)|yq|$ and $\nd_G(t, s) \leq (1 + \eps)\nd_{G_B}(t, s) = (1 + \eps)|ts|$. Then, we have:
	\begin{equation}
		(1 + \eps)|ty| > \nd_G(y, q) + \nd_G(q, s) + \nd_G(s, t)~.
	\end{equation}
	The edge $(t, y)$ cannot be in the shortest path from $s$ to $t$ in $G$ or the path from $y$ to $q$ because $|ty| > \max\{\nd_{G_B}(y, q), \nd_{G_B}(s, t)\}$. Since $(q, s) \in G$, we obtain $\nd_{G- (t, y)}(q, s) = \nd_G(q, s)$.  Hence, we have:
	\begin{equation}
		(1 + \eps)|ty| > \nd_{G - (t, y)}(y, q) + \nd_{G - (t, y)}(q, s) + \nd_{G - (t, y)}(s, t) \geq \nd_{G - (t, y)}(t, y)~.
	\end{equation}
	Thus, $(1 + \eps)\delta_{G_B}(t,y) > \nd_{G-(t, y)}(t, y)$, a contradiction.
\end{proof}

The proof of \Cref{lm:EuiclideanGreedyLanky} implies that $G$ is $O(\eps^{1 - 2d})$lanky. The same proof of \Cref{thm:strong_euclidean_separator} yields \Cref{thm:strong_unibal_separator}.

\subsection{Derandomization}\label{subsec:derandomize}
In this section, we give a deterministic algorithm for finding a sublinear separator of a set of points in $\mathbb{R}^d$ whose spread is bounded by a parameter $\Delta  > 0$. The same argument could be applied to point sets of low fractal dimension and unit ball graphs.  

\begin{theorem}
	\label{thm:deterministic_separator_greedy}
	Let $P$ be a given set of points in $\mathbb{R}^d$ with spread $\Delta$ and $G$ be the greedy $(1+\epsilon)$-spanner of $P$ for some $\eps \in (0,1/2]$. Then any $k$-vertex subgraph $H$ of $G$ has a $\left(1 - \frac{\eps^{O(d)}}{d}\right)$-balanced separator $S$ of size $O(\log{\Delta}^{1/d}\eps^{-O(1)}k^{1-1/d})$. Furthermore, $S$ can be found in $O(k)$ time given $H$.
\end{theorem}


Our algorithm uses a result of \cite{EMT95}. 

\begin{definition}
	A set of balls $\mathcal{B}$ in $\mathbb{R}^d$ is \emph{$\psi$-ply} if for any point $p \in \mathbb{R}^d$, there are at most $\psi$ balls in $\mathcal{B}$ containing $p$.  
\end{definition}

Given a set  $\psi$-ply set $\mathcal{B}$ of $n$ balls, we say that a ball $\ball(p, r)$ is an  \emph{$(\alpha,\beta)$-separating ball} if (a) the number of balls entirely contained in  $\ball(p, r)$ and the number of balls outside and disjoint from  $\ball(p, r)$ are at most $\alpha$ each and (b) the number of balls intersecting $\ball(p, r)$ is at most $\beta$.  Eppstein, Miller and Teng\cite{EMT95} gave a deterministic linear time algorithm that finds a balanced separating ball for any given  $\psi$-ply set of balls.

\begin{theorem}[Theorem 4.2 in \cite{EMT95}, adapted]
	\label{thm:deterministic_separator}
	Given a $\psi$-ply set $\mathcal{B}$ of $n$ balls in $\mathbb{R}^d$,  one can find a $\left(\frac{d + 1}{d + 2}n, O(\psi^{1/d}n^{1 - 1/d})\right)$-separating ball for  $\mathcal{B}$  in (deterministic) linear time. 
\end{theorem} 

Given a set of edges $E$, we define a set of balls $\mathcal{B}_E$ as follows. For each edge $e\in E$, we construct a ball $\ball_e = \ball(m_e,w(e)/2)$ with diameter $e$ where the center $m_e$ is the midpoint of $e$.  We then define $\mathcal{B}_E = \{\ball_e: e\in E\}$.  In \Cref{lm:ply} below, we show that, if a graph $G$ is a greedy spanner, then $\mathcal{B}_{E(G)}$ is $\left(\eps^{-O(d)}\log(\Delta)\right)$-ply. We later use this result and \Cref{thm:deterministic_separator} to find a separator of $G$ deterministically.

\begin{lemma}
	\label{lm:ply}
	Let $P$ be a given set of points in $\mathbb{R}^d$ with spread $\Delta$ and $G$ be the greedy $(1+\epsilon)$-spanner of $P$ for some $\eps \in (0,1/2]$.  
	Then $\mathcal{B}_{E(G)}$ is $(\eps ^{-O(d)}\log{\Delta})$-ply.
\end{lemma}

We now focus on proving \Cref{lm:ply}. 	Let $p$ be an arbitrary point in $\mathbb{R}^d$ and $\mathcal{B}_p$ be the set of balls in $\mathcal{B}_{E(G)}$ containing $p$. Our goal is to bound $\mathcal{B}_p$ by a function of $d$ and $\eps$. Observe that:

\begin{observation}\label{obs:me-vs-p} Given an edge $e = (u, v)$. If $p$ is in $\ball_e$, then $m_e\in \ball(p, |uv| / 2)$. 
\end{observation}
\begin{proof}
	Observe that $p$ is in $\ball_e$ if and only if $|pm_e| \leq \frac{|uv|}{2}$, which implies that $m_e$ is in $\ball(p, |uv| / 2)$.
\end{proof}
Let $d_{min}$ and $d_{max}$ respectively be the minimum and maximum distance between two points in $P$.  By \Cref{obs:me-vs-p}, every center of the balls in $\mathcal{B}_p$ must be inside $\ball(p, d_{max} / 2)$. Let $M_p = \{e\in E(G): \ball_e \in \mathcal{B}_p\}$. For each $i \in [1, \lfloor \log{\Delta} \rfloor + 1]$, let $M_p^i =\{e: e\in M_p, 2^{i - 1}d_{min} \leq w(e)  <   2^{i}d_{min} \}$. By definition, $M_p$ is the union of all $M_p^i$.  By \Cref{obs:me-vs-p}, we have:
\begin{observation}
	\label{obs:level_balls}
	The centers of the balls in $\mathcal{B}_{M_p^i}$ are in $\ball(p, 2^{i - 1}d_{min})$.
\end{observation}

Next, we cover $\ball(p, 2^{i - 1}d_{min})$ by a set of $O(\eps^{-d})$ balls of radius $(2^{i - 1}\eps d_{min}) / 24$, denoted by $\mathcal{N}_i$. Note that $\mathcal{N}_i$  exists due to \Cref{lm:packing}.  In the following claim, we show that each ball in $\mathcal{N}_i$ contains at most $\eps^{-O(d)}$ centers of  the balls in $\mathcal{B}_{M_p^i}$.
\begin{lemma}
	\label{clm:bounded_midpoints}
	For any point $x$ and a positive real number $D$, $\ball(x, \eps D/24)$ contains at most $\eps^{-O(d)}$ midpoints of  the edges in $E(G)$  whose lengths are between $D$ and $2D$. 
\end{lemma}
\begin{proof}
	Similar to the proof of \Cref{lm:XtoYEdges}, we cover the space $\mathbb{R}^d$ by a set of $O(\eps^{-d})$ cones with apex $x$ and angle $\eps/8$. Let $M_{x, D} = \{e \in E(G): D \leq w(e) \leq 2D,~m_e \in \ball(x, \eps D / 24)\}$ and $V(M_{x, D})$ is the set of endpoints of the edges in $M_{x, D}$. Our goal is to show that $|M_{x, D}| = \eps^{-O(d)}$. First, we show that each cone contains at most one point $V(M_{x, D})$ which is incident to exactly one edge in $M_{x, D}$. Assume that there is a cone containing two points $y$ and $q$ that are endpoints of two \emph{different} edges in $M_{x, D}$; $y$ and $q$ might be the same point. Let $t$ and $s$ be the points such that $(t, y), (s, q) \in M_{x, D}$. See \Cref{fig:cross_cone} for an illustration. W.l.o.g, we assume that $|xy| \geq |xq|$. By \Cref{lm:dist_in_small_cone}, we have $|xy| \geq |yq| + (1 - \eps/4)|xq|$. To obtain a contradiction, we show  below that $(1 + \eps)|ty| > \nd_{G-(t, y)}(t, y)$. 
	\begin{figure}[htb]
		\center{\includegraphics[width=0.5\textwidth]{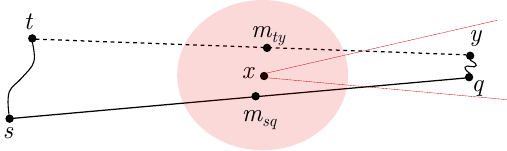}}
		\caption{A cone containing two points in $V(M_{x, D})$. We later prove that this case cannot happen.}
		\label{fig:cross_cone}
	\end{figure}

	Let $m_{ty}$ and $m_{sq}$ be the midpoints of $(t, y)$ and $(s, q)$, respectively. Because $m_{ty}$ and $m_{sq}$ are both in $\ball(x, \eps D / 24)$, we have $|xm_{ty}|, |xm_{sq}| \leq \eps D/24$.
	\begin{claim}
		\label{clm:mid_vector}
		$2\direct{m_{ty}m_{sq}} = \direct{ts} + \direct{yq}$.
	\end{claim}
	\begin{proof}
		Since $\direct{ts} + \direct{sy} + \direct{yq} + \direct{qt} = 0$, it holds that $\direct{ts} + \direct{yq} = \direct{ys} + \direct{tq}$. We have:
		\begin{equation}
		2\direct{m_{ty}m_{sq}} = \direct{m_{ty}s} + \direct{m_{ty}q} = \frac{\direct{ts} + \direct{ys} + \direct{tq} + \direct{yq}}{2} = \direct{ts} + \direct{yq} \qquad ,
		\end{equation}
		as claimed.
	\end{proof}
	Since $\direct{ts} = \direct{yq} - 2\direct{m_{ty}m_{sq}}$, it follows from the triangle inequality  that:
	\begin{equation}\label{eq:ts-vs-yq}
	|ts| \leq |yq| + 2|m_{ty}m_{sq}| \leq |yq| + \eps D/6.
	\end{equation} 
	Moreover, 
	\begin{equation}\label{eq:ts-vs-yqts}
	\begin{split}
	|ty| &= 2|ym_{ty}| \geq 2|xy| - 2|xm_{ty}| \geq 2|xy| - \eps D / 12  \qquad \mbox{(by the triangle inequality)}\\
	&\geq 2|yq| + 2(1 - \eps/4)|xq| - \eps D / 12 \qquad \mbox{(since $|xy| \geq |yq| + (1 - \eps/4)|xq|$ by \Cref{lm:dist_in_small_cone})}\\
	&\geq |yq| + |ts| - \eps D / 6 + 2(1 - \eps/4)|xq| - \eps D / 12  \qquad\mbox{(by \Cref{eq:ts-vs-yq})}\\
	&\geq |yq| + |ts| + 2(1 - \eps/4)(|m_{sq}q| - |m_{sq}x|) - \eps D/4  \qquad \mbox{(by the triangle inequality)}\\
	&\geq |yq| + |ts| + (1 - \eps / 4)|sq| - \eps D / 3  \qquad \mbox{(since $|m_{sq}x| \leq \eps D/24$ and $2|m_{sq}q| = |sq|$)}
	\end{split}
	\end{equation}
	Multiplying both sides of \Cref{eq:ts-vs-yqts} by $(1 + \eps)$, we obtain:
	\begin{equation}
	\label{eq:bound_ty}
	\begin{split}
	(1 + \eps)|ty| &\geq (1 + \eps)(|yq| + |ts| + (1 - \eps / 4)|sq| - \eps D / 3)\\
	&= (1 + \eps)|yq| + (1 + \eps)|ts| + |sq| + (3\eps/4 - \eps ^2/4)\underbrace{|sq|}_{\geq D} - (1 + \eps)\eps D/3\\ 
	&\geq (1 + \eps)|yq| + (1 + \eps)|ts| + |sq| + \eps D\underbrace{(3/4 - \eps /4 - (1 + \eps)/3)}_{\mbox{positive for all $\eps \in (0, 1/2]$}}\\
	&> (1 + \eps)|yq| + (1 + \eps)|ts| + \nd_{G - (t, y)}(s, q) \qquad \mbox{(since $G - (t, y)$ contains $(s, q)$)}
	\end{split}
	\end{equation}

	To bound the value of $(1 + \eps)|ty|$, we show in the following claim that the distances from   $y$ to $q$ and from $t$ to $s$ are unchanged after removing $(t, y)$ from $G$.
	\begin{claim}
		\label{clm:dist_after_remove_ty}
		$\nd_{G - (t, y)}(t, s) = \nd_G(t, s)$ and $\nd_{G - (t, y)}(y, q) = \nd_G(y, q)$.
	\end{claim}
	\begin{proof}
		We show that $|yq|,~|ts| < |ty|/(1 + \eps)$. By \Cref{lm:dist_in_small_cone}, we have:
		\begin{equation}
		\label{eq:bound_yq}
		\begin{split}
		|yq| &\leq |xy| - (1 - \eps/4)~|xq| \leq |xy| \leq |m_{ty}y| + |xm_{ty}| \qquad \mbox{(by the triangle inequality)}\\
		&\leq |ty|/2 + \underbrace{\eps D/24}_{\leq \eps |ty|/24} \leq (1/2 + \eps/24)~|ty|\\
		& < |ty|/(1 + \eps) \qquad \mbox{since $\eps \in (0, 1/2]$}
		\end{split}
		\end{equation}
		By \Cref{eq:ts-vs-yq}, $|ts| \leq |yq| + \eps D / 6$. Thus, by \Cref{eq:bound_yq}, it holds that:
		\begin{equation}\label{eq:bound_ts}
		|ts| \leq |ty| / 2 + \eps D/24 + \eps D /6 = |ty| / 2 + 5\eps D/24 \leq (1/2 + 5\eps / 24)~|ty| < |ty| / (1 + \eps)
		\end{equation}
		By \Cref{eq:bound_yq,eq:bound_ts},  $(1 + \eps)~|ts|, (1 + \eps)~|yq|$ are both less than $|ty|$. On the other hand, $\nd_{G}(t, s) \leq (1 + \eps)~|ts| < |ty|$ and $\nd_{G}(y, q) \leq (1 + \eps)~|yq| < |ty|$.  It follows that edge $(t, y)$ is not on the shortest path from $t$ to $s$ and from $y$ to $q$.  This means that $\nd_{G - (t, y)}(t, s) = \nd_G(t, s)$ and $\nd_{G - (t, y)}(y, q) = \nd_G(y, q)$ as desired. 
	\end{proof}
	Apply $\nd_{G - (t, y)}(y, q) = \nd_G(y, q) \leq (1 + \eps)~|yq|$ and $\nd_{G - (t, y)}(t, s) = \nd_G(t, s) \leq (1 + \eps)~|ts|$ to \Cref{eq:bound_ty}, we obtain:
	\begin{equation}
	(1 + \eps)~|ty| > \nd_{G - (t, y)}(y, q) + \nd_{G - (t, y)}(t, s) + \nd_{G - (t, y)}(s, q) \geq \nd_{G - (t, y)}(t, y).
	\end{equation}
	By \Cref{fact:path-greedy}, we obtain a contradiction. Thus, there is at most one point in $V(M_{x, D})$ in each cone which is incident to exactly one edge in $M_{x,D}$. Since there are $\eps^{-O(d)}$ cones, $|V(M_{x, D})| = \eps^{-O(d)}$. Therefore, $M_{x, D}$ contains at most $\eps ^{-O(d)}$ edges as claimed.
\end{proof}
We now finish the proof of \Cref{lm:ply}.
\begin{proof}[Proof of \Cref{lm:ply}]
	By \Cref{clm:bounded_midpoints}, each ball in $\mathcal{N}_i$ contains at most $\eps^{O(-d)}$ centers of the balls in $\mathcal{B}_{M_p^i}$. Because $\ball(p, 2^{i - 1}d_{min})$ is covered by $\mathcal{N}_i$, $\ball(p, 2^{i - 1}d_{min})$ contains at most $|\mathcal{N}_i|\eps^{-O(d)} = \eps^{-O(d)}$ centers of the balls in $\mathcal{B}_{M_p^i}$. By \Cref{obs:level_balls}, all centers of balls in $\mathcal{B}_{M_p^i}$ are in $\ball(p, 2^{i - 1}d_{min})$. Therefore, $|M_p^i| = |\mathcal{B}_{M_p^i}| = \eps^{-O(d)}$. Since $\mathcal{B}_p$ is the union of $\mathcal{B}_{M_p^i}$ for $i \in [1, \lfloor \log{\Delta} \rfloor + 1]$, $|\mathcal{B}_p| \leq  \sum_{i}|\mathcal{B}_{M_p^i}| = \eps^{-O(d)}\log{\Delta}$  as desired. 
\end{proof}
We now give a deterministic linear-time  algorithm to find a separator of a greedy spanner in Euclidean spaces as claimed in \Cref{thm:deterministic_separator_greedy}.

\begin{proof}[Proof of \Cref{thm:deterministic_separator_greedy}]
	We prove \Cref{thm:deterministic_separator_greedy} for the special case $H = G$. The argument for any subgraph $H$ follows the same line. Let $n$ be the number of vertices in $P$ and $E$ be the edge set of $G$. 
	
	If $n = O(d^{2d}\log{\Delta}~\eps^{-O(d)})$, then $\log{\Delta}^{1/d}\eps^{-O(1)}n^{1 - 1/d} = \Omega(n / d^2)$. Thus, if we let $S$ be any set of vertices of size $n/d^2$ of $G$, then $|S| = O(\log{\Delta}^{1/d}\eps^{-O(1)}n^{1 - 1/d})$. Furthermore, $S$ is a $\left(1 - \frac{1}{d^2}\right)$-balanced separator of $G$, \Cref{thm:deterministic_separator_greedy} follows from the fact that $1-1/d^2 \leq 1-\eps^d/d$ when $\eps \leq 1/2$. 
	
	Otherwise, $n = \Omega(d^{2d}\log{\Delta}~\eps^{-O(d)})$. By \Cref{lm:ply}, $\mathcal{B}_E$ is $(\eps^{-O(d)}\log{\Delta})$-ply. By  \Cref{thm:deterministic_separator}, we can find a $\left(\frac{d + 1}{d + 2}, O(\log{\Delta}^{1/d}\eps^{-O(1)}n^{1 - 1/d})\right)$-separating ball $\ball(p, r)$ of $\mathcal{B}_E$ in deterministic linear time. Let $\mathcal{B}_{in}, \mathcal{B}_{ext}$ and $\mathcal{B}_{cut}$ be the set of balls strictly contained  in the interior, exterior and cutting  $\ball(p, r)$. Similar to \Cref{thm:sublinearChar}, we construct a separator $S$ by picking all vertices in $\ball(p, r)$ that are incident to the edges cutting $\ball(p, r)$. By the construction of $\mathcal{B}_E$, we get $|S| \leq |\mathcal{B}_{cut}| = O(\log{\Delta}^{1/d}\eps^{-O(1)}n^{1 - 1/d})$. To complete the proof, we show that there are at least $\frac{\eps^{O(d)}n}{d}$ vertices outside $\ball(p, r)$. The bound for the number of inside vertices follows by a symmetric argument. The number of balls in the exterior of $\ball(p, r)$ is $|\mathcal{B}_{ext}| = |\mathcal{B}_E| - |\mathcal{B}_{in}| - |\mathcal{B}_{cut}| \geq |\mathcal{B}_E| - \frac{d + 1}{d + 2}|\mathcal{B}_E| - |\mathcal{B}_{cut}| = \frac{|\mathcal{B}_E|}{d + 2} - |\mathcal{B}_{cut}| $. Because $|\mathcal{B}_E| = |E| \geq n - 1$ and $n = \Omega(d^{2d}\log{\Delta}~\eps^{-O(d)})$, $|\mathcal{B}_{cut}| = O(\log{\Delta}^{1/d}\eps^{-O(1)}n^{1 - 1/d}) \leq \frac{|\mathcal{B}_E|}{(d + 2)(d + 3)}$. It follows that $|\mathcal{B}_{ext}| \geq \frac{|\mathcal{B}_E| }{d + 2}- \frac{|\mathcal{B}_E|}{(d + 2)(d + 3)} \geq \frac{|\mathcal{B}_E|}{d + 3}$.  That is, there are at least $\frac{|E|}{d + 3}$ edges in the exterior of $\ball(p, r)$. By \Cref{lm:EuiclideanGreedyLanky}, each vertex in $G$ has degree at most $\eps^{1 - 2d}$. Therefore, there are at least $\frac{|E|}{(d + 3)\eps^{1 - 2d}}$ vertices in the exterior of $\ball(p, r)$. Using the fact that $|E| \geq n - 1$, we conclude the number of points outside $\ball(p, r)$ is at least $\frac{n\eps^{O(d)}}{d}$. By the same argument, the number of points inside $\ball(p, r)$ is at least $\frac{n\eps^{O(d)}}{d}$. Therefore, $S$ is a $\left(1 - \frac{\eps^{O(d)}}{d}\right)$-balanced separator of $G$.  
\end{proof}
\section{Bounded Degree Spanners with Sublinear Separators in Doubling Metrics}

Chan, Gupta, Maggs and Zhou~\cite{CGMZ05}, hereafter CGMZ, constructed a $(1+\eps)$-spanner with a maximum degree of $\eps^{-O(d)}$ for points in doubling metrics of dimension $d$. We show in this section that their spanner is $\tau$-lanky for $\tau = \eps^{-O(d)}$. This implies that CGMZ's spanner has sublinear separator, thereby resolving the question asked by Abam and Hal-Peled in ~\cite{AH10}.

We assume doubling metric $(X,\nd_X)$ has dimension $d$, minimum pairwise distance  $1$ and spread $\Delta$.  CGMZ's construction relies on a net-tree that we define below.

\begin{definition}[Net Tree]\label{def:nettree}  Let $r_0 = 1/4$ and $N_0 = X$. For each integer $i \in [1, \lceil \log(\Delta)\rceil + 2]$, we define $r_i = 2^{i}r_0$ and $N_i \subseteq N_{i-1}$ be the $r_i$-net of $N_{i-1}$.  The hierarchy of nets $N_0 \supseteq N_1 \supseteq \ldots\supseteq N_{\lceil\log(\Delta) \rceil + 2} $ induces a rooted tree $T$ where:
	\begin{itemize}[noitemsep]
		\item[(1)] $T$ has $\lceil\log(\Delta) \rceil + 3$ levels in which level $i$ corresponds to the set of points $N_i$ for every $i \in [0,\lceil \log(\Delta)\rceil + 2]$.
		\item[(2)] The parent of a vertex $v$ at level $i \in [0,\lceil \log(\Delta)\rceil + 1]$ is the closest vertex in $N_{i+1}$; ties are broken arbitrarily.
	\end{itemize}
\end{definition}
Note that since $\Delta$ is the maximum pairwise distance, $ N_{\lceil\log(\Delta) \rceil + 2}$ contains a single point. Additionally, the starting radius $r_0= 1/4$ is an artifact introduced by  CGMZ~\cite{CGMZ05} to handle some edge cases gracefully; one could construct the net tree with the starting radius $r_0 = 1$. 	

Since a point $x \in X$ may appear in many different nets, we sometimes write $(x,i)$ to explicitly indicate the copy of $x$ in the net $N_i$. We denote by $p(x,i)$ the parent of the point $(x,i)$, which is a point in $N_{i+1}$ when $(x,i)$ is not the root of the tree $T$. For each point $x\in X$, we denote by $i^*(x)$ the maximum level that $x$ appears in the net tree $T$.
	
\paragraph{CGMZ Spanner Construction.~} The spanner construction of CGMZ has two steps. In the first step, they construct a $(1+\eps)$-spanner $G_1$ with $O(n)$ edges using the \emph{cross edge rule}. Edges of $G_1$ are then oriented by the maximum level of the endpoints: an edge $(u,v)$ is oriented as $(u\rightarrow v)$ if $i^*(u) < i^*(v)$; if both endpoints have the same maximum level, then the edge is oriented arbitrarily. CGMZ showed that the out-degree of every vertex of $G_1$ after the orientation is a constant. In the second step, they reduced the in-degree of every vertex by rerouting to its neighbors following a \emph{rerouting rule}. The graph after rerouting edges has bounded degree (for constant $\eps$ and $d$).  In the following, we formally describe two steps. 
 
 \begin{tcolorbox}
 	\hypertarget{CGMZ}{}
 	\textbf{CGMZ Algorithm~\cite{CGMZ05}:} The construction has two steps.
 		\begin{itemize}
 			\item \textbf{Step 1:~} Let $\gamma = 4 + \frac{32}{\eps}$. For each net $N_i$ at level $i$ of the net tree $T$, let $E_i$ be the set of all pairs $(u, v)$ of vertices in $N_i$ such that $\nd_X(u, v) \leq \gamma r_i$ and $(u, v)$ has not appeared in $\bigcup_{j=0}^{i-1}E_{j}$. We call $E_i$ the set of \emph{cross edges} at level $i$. The edge set of $G_1$ is the union of all $E_i$'s:  $E(G_1) = \bigcup_{i=0}^{\lceil \log(\Delta)\rceil + 2} E_i$. Let $\direct{G_1}$ be obtained from $G_1$ by orienting every edge $(u,v)$ as $(u\rightarrow v)$ if $i^*(u) < i^*(v)$, as $(v\rightarrow u)$ if $i^*(u) > i^*(v)$, and arbitrarily if $i^*(u) = i^*(v)$. Let $\direct{E}_i$ be the set of oriented edges corresponding to $E_i$. 
 			\item \textbf{Step 2:~} Let $l = \lceil (1/\eps) \rceil + 1$. For each point $w \in N_i$, $i \in [0,\lceil \log(\Delta)\rceil + 2]$, $N^{in}_i(w) = \{v \in N_i: (v\rightarrow w) \in E(\direct{G_1})\}$ be the set of \emph{in-neighbors at level $i$ of $w$}.  We construct another directed graph $\direct{G_2}$ with the same vertex set $X$ as follows. Initially, $\direct{G_2}$ has no edge. Let  $I_w = \{i_1, i_2, \ldots, i_{m_w}\}$ be the list of levels sorted by increasing order such that $w$ has \emph{at least one} in-neighbor at each level in the list; $m_w \leq i^*(w)$. For each index $j \leq [1,m_w]$, if $j \leq l$, we add every directed edge in $\direct{E}_{i_j}$ to $\direct{G_2}$; otherwise, $j > l$,  we choose a point $u \in N_{j-l}$ and for every edge $(v \rightarrow w)\in \direct{E}_{i_j}$, we add an edge $(v \rightarrow u)$ to $\direct{G}_2$. We say that $(v\rightarrow w)$  is \emph{rerouted} to $u$. Let $G_2$ be obtained from  $\direct{G_2}$ by ignoring the directions of the edges; we return $G_2$  as the output spanner.
 		\end{itemize}
 \end{tcolorbox}

Our goal is to show that $G_2$ is $\tau$-lankly for $\tau = \eps^{-O(d)}$. To this end, we show several properties of $G_1$ and $G_2$.  CGMZ~\cite{CGMZ16} showed the following lemmas.

\begin{lemma}\label{lm:G1Prop} $G_1$ and $\direct{G}_1$ have following properties: 
	\begin{itemize}[noitemsep]
		\item[(1)] For every point $u\in N_i$, $i \in [0,\lceil \log(\Delta)\rceil + 2]$,  $|N_i^{in}(u)| \leq (4\gamma)^d$; Claim 5.9~\cite{CGMZ16}.
		\item[(2)] The out-degree of a vertex $u$ in $\direct{G_1}$ is $\eps^{-O(d)}$; see Lemma 5.10~\cite{CGMZ16}.
	\end{itemize}
\end{lemma}

The proof of Item (1) in \Cref{lm:G1Prop} is by using the packing bound (\Cref{lm:packing}). For Item (2) in \Cref{lm:G1Prop}, CGMZ observed that that for every vertex $u$, there are  $O(\log(1/\eps))$ levels where $u$  could have edges oriented out from $u$; the bound on the out-degree then follows from Item (1).

\begin{lemma}[Lemma 5.12~\cite{CGMZ16}]
	\label{lm:boundedDeg} $G_2$ is a $(1+4\epsilon)$-spanner of $(X,\nd_X)$ with maximum degree $\eps^{-O(d)}$.
\end{lemma}

CGMZ used the following claim to show the stretch bound $(1+4\eps)$ of $G_2$ (Equation (4)~\cite{CGMZ16}). Since we also use this claim here, we include its proof for completeness.

\begin{claim} \label{clm:reroutedLength}
	If $(v \rightarrow w)$ is rerouted to $(v \rightarrow u)$, then  $\nd_X(u, w) \leq \eps \nd_X(v, w)$ and $\nd_X(v, w) \geq \nd_X(v, u) / (1 + \eps)$.
\end{claim}
\begin{proof}
	First, we show that $\nd_X(u, w) \leq \eps\nd_X(v, w)$.  Let $i$ and $j$ be two levels such that $(v, w) \in E_i$ and $(u, w) \in E_j$. It follows that $\nd_X(v, w)\geq \gamma r_{i - 1}$ since $(v,w)$ does not appear in $E_{i-1}$, and that  $\nd_X(u, w) \leq \gamma r_{j}$ by the definition of $E_j$. By the construction in Step 2, $i\geq j + l$. We have that $\nd_X(v, w) \geq \gamma r_{i - 1} \geq 2^l\gamma r_{j} \geq \frac{\nd_X(u, w)}{\eps}$. Thus, $\nd_X(u, w)\leq \eps \nd_X(v,w)$ as claimed. 
	
	By the triangle inequality, $\nd_X(v, u) \leq \nd_X(v, w) + \nd_X(w, u) < (1 + \eps)\nd_X(v, w)$; the second claim follows.
\end{proof}

Next, we show that any ball of radius $r$ only contains exactly one point that is incident to a long edge in $G_1$. Recall that $\gamma = 4 + \frac{32}{\eps}$ is the parameter defined in \hyperlink{CGMZ}{CGMZ Algorithm}.

\begin{lemma}\label{lm:longEdgeEndpts} Let $\ball(p,r)$ be any ball of radius $r$ centered at some point $p$.  There is at most one point in $\ball(p, r)$ that is incident to an edge in $G_1$ with length at least $4\gamma r$. \\
	More generally, for any parameter $\beta > 0$, there are at most $2^{O(d)}\left(\frac{\gamma}{\beta}\right)^d$ points in $\ball(p, r)$ that are incident to edges in $G_1$ of length at least $\beta r$. 
\end{lemma}
\begin{proof}
	By the construction of $G_1$, two endpoints of an edge with length at least $4\gamma r$ must be in $r_i$-net $N_i$ such that $r_i \geq 4r$. Since any two points in $N_i$ has distance at least $r_i\geq 4r$, $|\ball(p,r)\cap N_i|\leq 1$; this implies the first claim.
	
	We now show the second claim. Let $\mathcal{B}$ be a set of balls obtained by taking balls of radius $\frac{\beta r}{4\gamma}$ centered at points in a $\left(\frac{\beta r}{4\gamma}\right)$-net of $\ball(p, r)$. By \Cref{lm:packing}, $|\mathcal{B}| = 2^{O(d)}\left(\frac{\gamma}{\beta}\right)^{d}$. Let  $E_\text{long}$ be the set of edges in $G_1$ of length at least $\beta r$. For each ball $B \in \mathcal{B}$, there is at most one point in $B \cap \ball(p, r)$ that is incident to an edge in $E_\text{long}$ by the first claim. It follows that the total number of points in $\ball(p, r)$ incident to an edge in $E_\text{long}$ is at most $|\mathcal{B}| = 2^{O(d)}\left(\frac{\gamma}{\beta}\right)^{d}$.
\end{proof}

Next, we show that there is a small number of edges in $G_2$ between any well-separated pair. For any two sets $A$ and $B$, we define $r_{AB} = \max\{\diam(A),\diam(B)\}$.

\begin{lemma}
	\label{lm:GuptaWSP}
	Let $(A, B)$ be a $\gamma$-separated pair in $(X,\nd_X)$. Then, there are at most $\eps^{-O(d)}$ edges in $G_2$ between $A$ and $B$.
\end{lemma}
\begin{proof}
  	We count the number of directed edges with length at least $\gamma r_{AB}$ betweem $A$ and $B$ in $\direct{G_2}$. (Recall that $\gamma = 4 + \frac{32}{\eps}$.) We consider the edges that are directed from $A$ to $B$; those directed from $B$ to $A$ can be counted by the same argument. By \Cref{clm:reroutedLength}, an edge, say $(v\rightarrow u)$, of length at least $\gamma r_{AB}$ from a point $v$ in $A$ must be rerouted from an edge $(v \rightarrow w) \in E(\direct{G_1})$ of length at least $\frac{\gamma r_{AB}}{1 + \eps}\geq \frac{\gamma\diam(A)}{1 + \eps}$. It follows that the set of the \emph{endpoints} of the directed edges of length at least $\gamma r_{AB}$  from $A$ to $B$  is a subset of the set of the endpoints in $A$, denoted by $U_A$, of edges of length at least  $\frac{\gamma r_{AB}}{1+\eps}$ in $G_1$. 
	
	By  applying \Cref{lm:longEdgeEndpts} to a minimum ball enclosing $A$ (of radius at most $\diam(A)$), we have that $|U_A|  = 2^{O(d)}\left(\frac{\gamma}{\gamma/(1+\eps)}\right)^d = 2^{O(d)}(1+\eps)^d = 2^{O(d)}$.   By \Cref{lm:boundedDeg}, there are at  most $|U_A|\eps^{-O(d)} = 2^{O(d)}\eps^{-O(d)} = \eps^{-O(d)}$ edges in $G_2$ incident to points in $U_A$; here we use the fact that $\eps \leq 1/2$.  It follows that the number of directed edges from $A$ to $B$ in $\direct{G_2}$ of length at least $\gamma r_{AB}$ is  $\eps^{-O(d)}$. By the same argument,  there are $\eps^{-O(d)}$ edges from $B$ to $A$ in $\direct{G_2}$ of length at least $\gamma r_{AB}$. The lemma now follows.
\end{proof}

We obtain the following generalization of \Cref{lm:GuptaWSP}. 

\begin{corollary}
	\label{col:GuptaWSP}
	Let $(A, B)$ be a $\beta$-separated pair in $(X,\nd_X)$. Then there are at most $\left(\frac{\gamma}{\beta}\right)^{2d}\eps^{-O(d)}$ edges in $G_2$ between $A$ and $B$.
\end{corollary}

\begin{proof}
 	We have that $\nd_X(A, B) \geq \beta r_{AB}$. Let $\mathcal{A}$ and $\mathcal{B}$ be minimum collections of balls with radius $\frac{\beta r_{AB}}{2\gamma}$ that cover $A$ and $B$, respectively.  By \Cref{lm:packing}, we have that $|\mathcal{A}|= 2^{O(d)}\left(\frac{2\diam(A)}{\beta r_{AB} / (2\gamma)}\right)^d \leq 2^{O(d)}\left(\frac{2r_{AB}}{\beta r_{AB} / (2\gamma)}\right)^d = 2^{O(d)}\left(\frac{\gamma}{\beta}\right)^d$. By the same argument, $|\mathcal{B}| = 2^{O(d)}\left(\frac{\gamma}{\beta}\right)^d$.
   	Let $(A', B')\in \mathcal{A}\times\mathcal{B}$ be any pair of balls. Observe that  $(A' \cap A, B' \cap B)$ is a $\gamma$-separated pair. By \Cref{lm:GuptaWSP}, there are  $\eps^{-O(d)}$ edges between $A' \cap A$ and $B' \cap B$. It follows that the number of edges between $A$ and $B$ in $G_2$ is $|\mathcal{A}| |\mathcal{B}|\eps^{-O(d)} =2^{O(d)} \left(\frac{\gamma}{\beta}\right)^{2d}\eps^{-O(d)} = \left(\frac{\gamma}{\beta}\right)^{2d}\eps^{-O(d)}$ since $\eps \leq 1/2$.
\end{proof}

Next, we show that that the number of short edges, those of length at least $r$ and at most $\gamma r$, cut by a ball with radius $r$ is small. This reduce the task of showing lankiness of $G_2$ to bounding the number of long edges of $G$. 

\begin{lemma}
	\label{lm:boundedShortEdges}
	 Let $\ball(p,r)$ be any ball of radius $r$ centered at some point $p$. The number of edges with length from $r$ to $\gamma r$ cut by $\ball(p,r)$ is $\eps^{-O(d)}$.
\end{lemma}
\begin{proof}
	Let $E_\text{short} \subseteq E(G_2)$ be the set of edges of length in $[r, \gamma r]$ cut by $\ball(p, r)$. Observe by the triangle inequality that for every edge $(u, v) \in E_\text{short}$, both $u$ and $v$ are in $\ball(p, (\gamma + 1)r)$. Let $\mathcal{B}_\text{in}$ ($\mathcal{B}_\text{out}$, resp.) be the minimum collection of balls with radius $\eps r$ that covers $\ball(p, r)$ ($\ball(p, (\gamma + 1)r)$, resp.). By the packing lemma (\Cref{lm:packing}), $|\mathcal{B}_\text{in}^i| = 2^{O(d)}\left(\frac{r}{\eps r}\right)^d = \eps^{-O(d)}$, and $|\mathcal{B}_\text{out}| = 2^{O(d)}\left(\frac{\gamma + 1}{\eps}\right)^d = \eps^{-O(d)}$.
	
	Let $(B_{\text{in}}, B_{\text{out}}) \in \mathcal{B}_\text{in} \times \mathcal{B}_\text{out}$ be any pair in of ball such that there is an edge $e = (u, v)$ between $B_{\text{in}}$ and $B_{\text{out}}$. By the triangle inequality, we have that $\nd_X(B_{\text{in}}, B_{\text{out}}) \geq \nd_X(u, v) - \diam(B_{\text{in}}) - \diam(B_{\text{out}}) \geq r - 4\eps r$. Thus,  $\frac{\nd_X(B_{\text{in}}, B_{\text{out}})}{r_{B_{\text{in}}, B_{\text{out}}}} \geq \frac{r - 4\eps r}{2\eps r} = \frac{1 - 4\eps}{2\eps}$. (Recall that $r_{B_{\text{in}}, B_{\text{out}}}$ by definition is the maximum diameter of $B_{\text{in}}$ and $B_{\text{out}}$.) This implies that $(B_{\text{in}}, B_{\text{out}})$ is a $\left(\frac{1 - 4\eps}{2\eps}\right)$-separated pair. By applying \Cref{col:GuptaWSP} with $\beta = \frac{1 - 4\eps}{2\eps}$, it follows that there are  $\left(\frac{2\gamma\eps}{1 - 4\eps}\right)^{2d}\eps^{-O(d)} = 2^{O(d)}\eps^{-O(d)} = \eps^{-O(d)}$ edges in $E_\text{short}$ between $B_{\text{in}}$ and $B_{\text{out}}$. Thus, $|E_\text{short}| = \eps^{-O(d)}|\mathcal{B}_\text{in} \times \mathcal{B}_\text{out}| = \eps^{-O(d)}$, as desired.
\end{proof}

\begin{figure}[htb]
	\center{\includegraphics[width=0.5\textwidth]{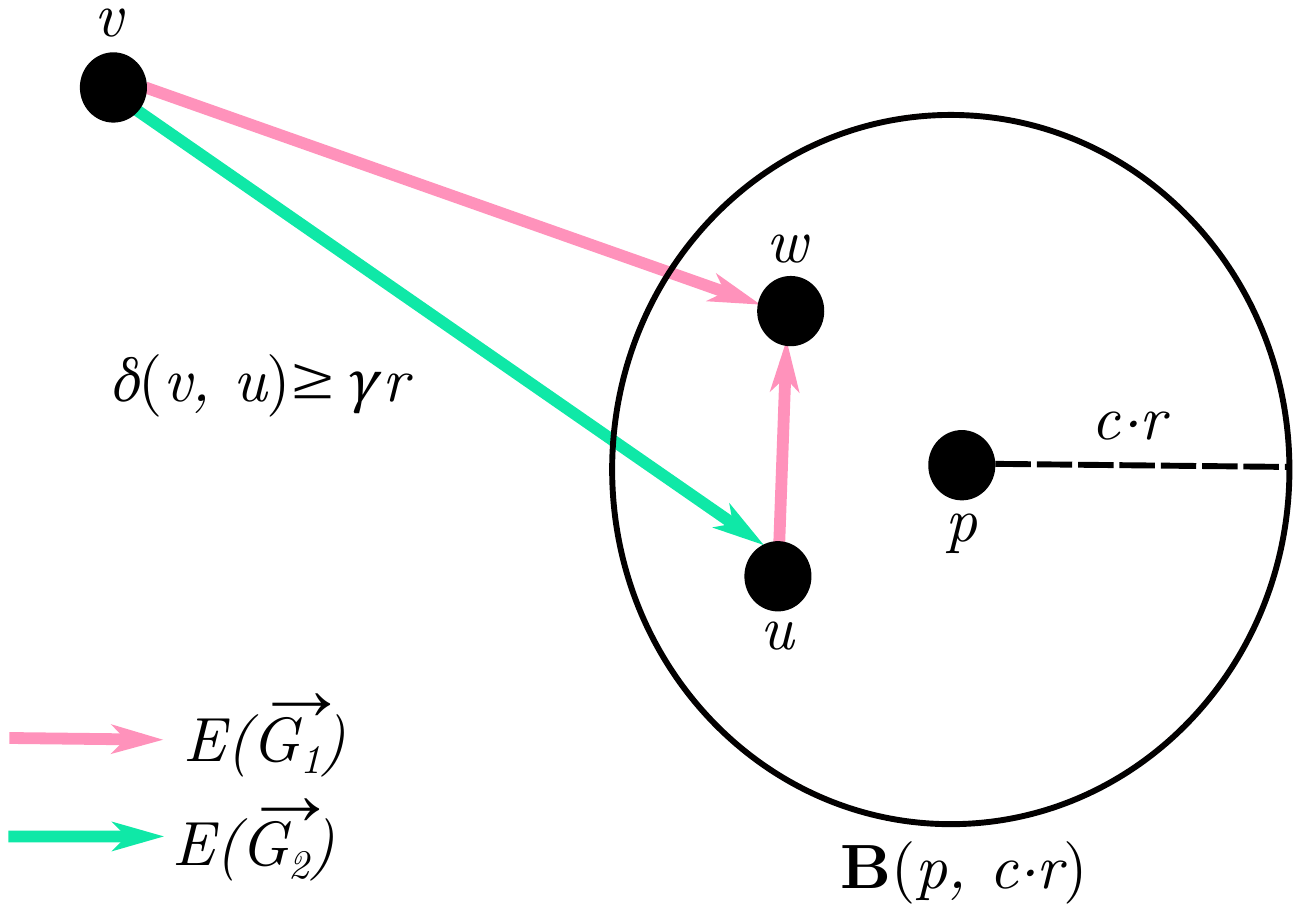}}
	\caption{An edge $(v \rightarrow w) \in E(\protect\direct{G_1})$ incident to a vertex $w\in \ball(p,c\cdot r)$ is rerouted to some point $u\in \ball(p,c\cdot r)$. That is, $(v \rightarrow u) \in E(\protect\direct{G_2})$.}
	\label{fig:reroute_to_long_edge}
\end{figure}

By \Cref{lm:boundedShortEdges}, to show that $G_2$ is lanky, it suffices to focus on edges of length at least $\gamma r$.  To that end, we show the following structure lemma, which bounds the number of directed edges in $\direct{G_1}$ of length at least $\gamma r$ that are incident the same vertex $w$ in a ball of radius $c\cdot r$ for some parameter $c\geq 1$ and rerouted (from $w$) to other points in the same ball in the construction of $\direct{G_2}$. See \Cref{fig:reroute_to_long_edge} for an illustration.

\begin{lemma}
	\label{lm:reroutedInBallLargeEdges}
 Let $\ball(p,c\cdot r)$ be any ball centered at some point $p$ that has radius $c \cdot r > 0$ for some parameter $c \geq 1$. For every point $w \in \ball(p, c\cdot r)$, there are $\eps^{-O(d)}(\lceil \log{c} \rceil + 1)$ edges $(v\rightarrow w)$ in $\direct{G_1}$ such that $(v\rightarrow w)$ is rerouted to a point $u$ in $\ball(p,c\cdot r)$ and $\nd_X(v, u) \geq \gamma r$.
\end{lemma}

\begin{proof}
	Let $I_w = \{i_1, i_2, \ldots\ i_{m_w}\}$ be the increasing sequence of indices such that $N^{in}_{i_t}(w) \not=\emptyset$ for every $1 \leq t \leq m_w$. Let $P_ \text{out}(w)$ be the set of vertices $v$ such that $(v \rightarrow w) \in E(\direct{G_1})$, $(v \rightarrow w)$ is rerouted to a point, say $u$, in $\ball(p,c \cdot r)$ and $\nd_X(v, u) \geq \gamma r$. Observe that $P_\text{out}(w)$ is a subset of $\bigcup_{t = 1}^{m_w}N^{in}_{i_{t}}(w)$. Let $i_j$ be the smallest index in $I_w$ such that $N^{in}_{i_j}(w) \cap P_\text{out}(w) \neq \emptyset$. Our goal is to show that $P_\text{out}(w) \subseteq \bigcup_{t = 0}^{l + \lceil \log{c} \rceil + 2}N^{in}_{i_{j + t}}(w)$. Recall that $l = \lceil (1/\eps) \rceil + 1$, which is defined in Step 2 of \hyperlink{CGMZ}{CGMZ Algorithm}.
	
	Suppose that there exists an integer $h \geq l + \lceil  \log{c} \rceil + 3$ such that $N^{in}_{i_{j + h}}(w) \cap P_\text{out}(w) \neq \emptyset$. Let $v_1\not= v_2$ be any two points in $P_{\text{out}}(w)$ such that $v_1 \in N^{in}_{i_j}(w)$ and $v_2 \in N^{in}_{i_{j + h}}(w)$. 
	Let $u_a$ be the point such that $(v_a  \rightarrow w)$ is rerouted to, $a\in\{1,2\}$. (See \Cref{fig:net_tree_rerouted_to_long_edge} for an illustration.) By the definition of $P_\text{out}(w)$, we have that $u_1, u_2 \in \ball(p, c\cdot r)$ and that $\nd_X(v_1, u_1), \nd_X(v_2, u_2) \geq \gamma r$. By \Cref{clm:reroutedLength}, $\nd_X(v_1, w) \geq \frac{\nd_X(v_1, u_1)}{1 + \eps} \geq \frac{\gamma r}{1 + \eps}$. Hence, $r_{i_j} \geq \frac{r}{1 + \eps}$ by the construction in Step 1 of \hyperlink{CGMZ}{CGMZ Algorithm}. Since $v_2 \in N^{in}_{i_{j + h}}(w)$, by the construction in Step 2 of \hyperlink{CGMZ}{CGMZ Algorithm}, $u_2 \in N^{in}_{i_{j + h-l}}(w)$. It follows that:
	\begin{equation}\label{eq:dist-wu2}
		\nd_X(w, u_2) \geq r_{i_{j + h - l}} \geq r_{i_j + h - l} \geq 2^{\lceil \log{c} \rceil  + 3}r_{i_j} \geq \frac{8c\cdot r}{1 + \eps} > 2c\cdot r,
	\end{equation}
	 since $\eps \leq 1$ and $c\cdot r> 0$. However, by the triangle inequality,  $\nd_X(w, u_2) \leq \nd_X(w, p) + \nd_X(p, u_2) \leq 2cr$, which contradicts \Cref{eq:dist-wu2}. Thus, $P_\text{out}(w)$ only contains points in $\bigcup_{t = 0}^{l + \log{c} + 2}N^{in}_{i_{j + t}}(w)$. 
	 
	 By  Item (1) of \Cref{lm:G1Prop}, there are at most $(4\gamma)^d$ points in each net that are adjacent to $w$ in $G_1$. Hence, $|P_{out}(w)| \leq (l + \lceil \log{c} \rceil + 3)O(4\gamma)^d = \eps^{-O(d)}(\lceil \log{c} \rceil + 1)$; this implies the lemma.
\end{proof}

\begin{figure}[htb]
	\center{\includegraphics[width=1\textwidth]{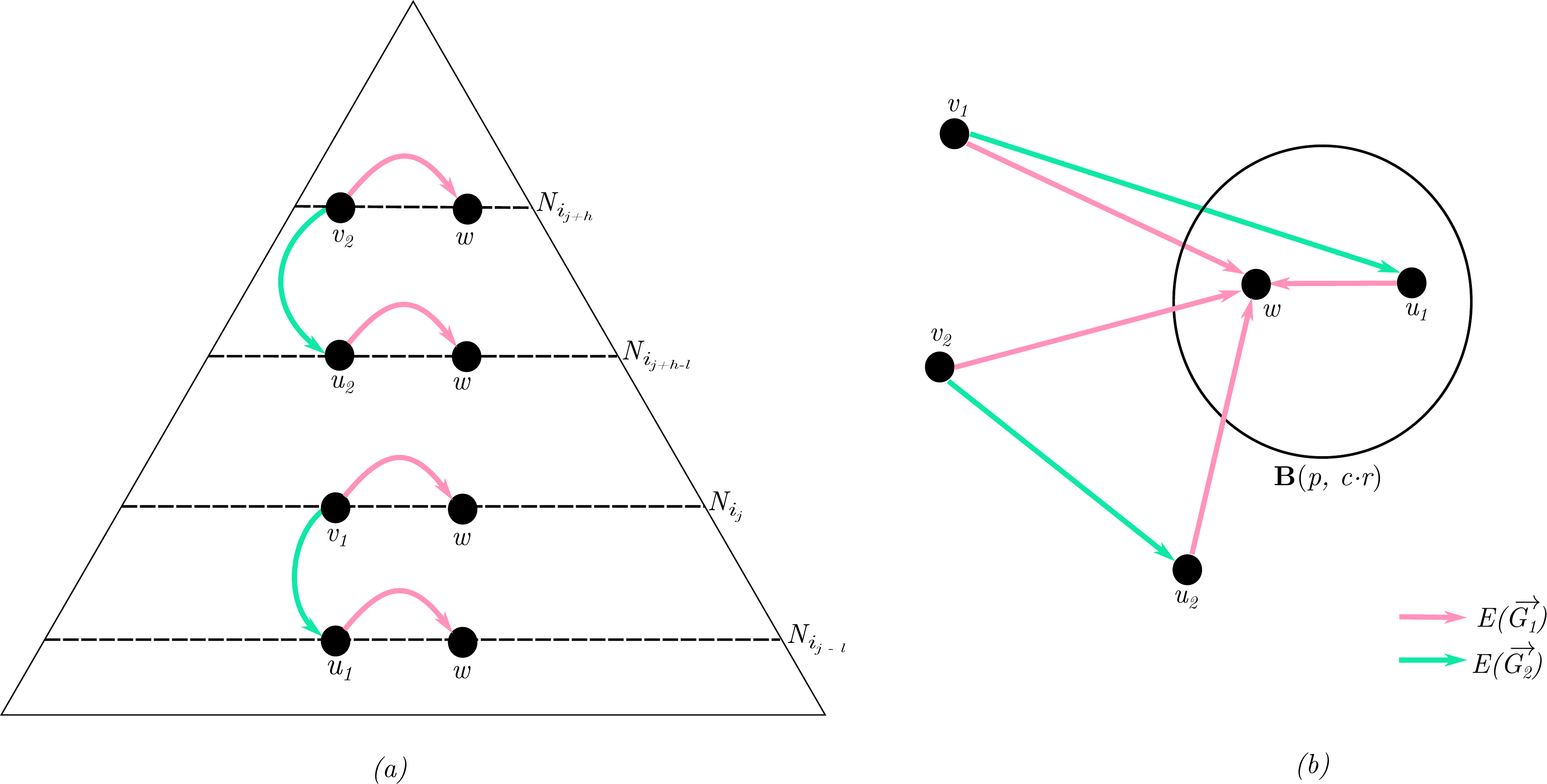}}
	\caption{$(a)$: The positions of $u_1, u_2, v_1, v_2$ and $w$ in the net tree and $(b)$ the edges $(v_1 \rightarrow w)$ and $(v_2 \rightarrow w)$ are rerouted to $u_1$ and $u_2$, respectively. We show in \Cref{lm:reroutedInBallLargeEdges} that $u_2$ must be outside $\ball(p, c \cdot r)$.}
	\label{fig:net_tree_rerouted_to_long_edge}
\end{figure}

 We now have all necessary tools to show that $G_2$ is lanky.

\begin{lemma}\label{lm:G2Lanky} $G_2$ is $\eps^{-O(d)}$-lanky.
\end{lemma}
\begin{proof}
	Let $\ball(p, r)$ be a ball of radius $r$. We assume that $r \geq 1/2$ since otherwise, we could enlarge $r$ to $1/2$ without changing the number of edges that cut $\ball(p,r)$. We show that the number of edges of length at least $r$ cut by $\ball(p, r)$ is bounded by $\eps^{-O(d)}$, which implies the lemma by the definition (\Cref{def:remotelyBoundedDeg}). 
	
	Let $E_\text{short}$ be the set of edges of length in $[r,\gamma r]$  cut by $\ball(p, r)$. Observe by  \Cref{lm:boundedShortEdges} that $|E_{\text{short}}| = \eps^{-O(d)}$. Thus, it remains to focus on the set of edges in $G_2$ of length at least $\gamma r$ cut by $\ball(p,r)$, which we denote by $E_\text{long}$. Let $\direct{E}_\text{long}$ be the set of directed edges corresponding to $E_\text{long}$ in $\direct{G_2}$.
	
	Recall that in the construction of Step 2 of \hyperlink{CGMZ}{CGMZ Algorithm}, some edges in $\direct{G_1}$ are copied over $\direct{G_2}$ (the case where $j\geq l$) while other edges are rerouted to different endpoints (the case where $j > l$). Let $\direct{E}^{1}_{\text{long}}$ be the set of edges in $\direct{G_2}$ that are copied from $\direct{G_1}$, and $\direct{E}^{2}_{\text{long}} = \direct{E}_{\text{long}}\setminus \direct{E}^{1}_{\text{long}}$. Let  $E^{1}_{\text{long}}$ and $E^{2}_{\text{long}}$ be the undirected counterparts of $\direct{E}^{1}_{\text{long}}$ and $\direct{E}^{2}_{\text{long}}$, respectively.
	
	We observe by the definition that if a point is incident to an edge in $E^{1}_{\text{long}}$, it is also incident to an edge with length at least $\gamma r$ in $G_1$. By applying  \Cref{lm:longEdgeEndpts} with $\beta = 1$, there are $2^{O(d)}$ points in $\ball(p, r)$ incident to an edge with length at least $\gamma r$ in $E(G_1)$. By \Cref{lm:boundedDeg}, each of these points is incident to at most $\eps^{-O(d)}$ edges in $E_\text{long}$ since $E_\text{long} \subseteq E(\direct{G}_2)$. Thus, \hypertarget{E1Long}{}$|E^{1}_{\text{long}}| = 2^{O(d)}\eps^{-O(d)} = \eps^{-O(d)}$ since $\eps \leq 1/2$. 
	
	 To bound the size of $E^{2}_{\text{long}}$, we partition $\direct{E}^{2}_{\text{long}}$ into two sets: $\direct{E}^{2}_\text{in} = \{(v \rightarrow u) \in \direct{E}^{2}_{\text{long}} ~|~  u \in \ball(p, r)\} $ is the set of edges directed into $\ball(p,r)$ and $\direct{E}^{2}_\text{out} = \{(u \rightarrow y)  \in \direct{E}^{2}_{\text{long}} ~|~  u \in \ball(p, r)\}$ is the set of edges directed out of $\ball(p,r)$. We bound the number of edges of $\direct{E}^{2}_\text{in}$ and $\direct{E}^{2}_\text{out}$ in \Cref{clm:E2in} and \Cref{clm:E2out}, respectively.
	
	\begin{claim}\label{clm:E2in}
		$|\direct{E}^{2}_\text{in}| = \eps^{-O(d)}$.
	\end{claim}
	\begin{proof}
		Let $j$ be the index that $4r \geq r_j > 2r$. Recall that for each $i \in [0,\lceil \log(\Delta)\rceil + 2]$, $E_i$ is the set of cross edge at level $i$ defined in Step 1. Let $\direct{M}_1$ be the set of edges $(v\rightarrow u)\in \direct{E}^{2}_\text{in}$ such that (a) $(v \rightarrow u)$ is formed by rerouting $(v \rightarrow w)$ to $u$ in Step 2 and (b) $(u,w)\in E_{j'}$ for some $j'\geq j$. Let $\direct{{M}_2} = \direct{E}^{2}_\text{in}\setminus \direct{M}_1$. 
		
		We first bound $|\direct{M}_1|$ by showing that every edge in $\direct{M}_1$ is directed towards the same endpoint, say $u$. If so,  $|\direct{M}_1| = \eps^{-O(d)}$ since $u$ has a degree  at most $\eps^{-O(d)}$ by \Cref{lm:boundedDeg}. Suppose for contradiction that there is another edge $(v' \rightarrow u')\in \direct{M}_1$ such that $u'\not= u$. We observe by the definition of cross edges in Step 1 that $u$ and $u'$ belongs to nets of level at least $j$. That means they both belong to the net $N_j$ at level $j$ since nets are nested. Thus, $\nd_X(u, u') \geq r_j > 2r$ by the definition of $j$. It follows that $u$ and $u'$ could not both be in $\ball(p, r)$; a contradiction.

		It remains to bound $|\direct{M}_2|$. To this end, we define $W = \{w ~|~ \exists (v \rightarrow u) \in \direct{M}_2 \text{ s.t. }  (v \rightarrow w) \text{ is rerouted to } u\}$. We claim that $W \subseteq \ball(p, (4\gamma + 1)r)$. Let $w$ be an arbitrary point in $W$. Let $u, v$ be the points that $(v \rightarrow w)$ is rerouted to $u$. Let $j'$ be the index such that $(u, w) \in E_{j'}$. By the definition of $\direct{M}_2$, $j' < j$. Since $(u, w) \in E_{j'}$, $\nd_X(u, w) \leq \gamma r_j \leq 4\gamma r$. By the triangle inequality, $\nd_X(p, w) \leq \nd_X(p, u) + \nd_X(u, w) \leq (4\gamma + 1)r$. Thus, $w \in \ball(p, (4\gamma + 1)r)$, as claimed. 
		
		By \Cref{clm:reroutedLength}, $\nd_X(v, w) \geq \frac{\nd_X(v, u)}{1 + \eps} \geq \frac{\gamma r}{1 + \eps}$. It follows that $w \in N_{h}$ where $h$ is the minimum index such that $r_h \geq \frac{r}{1 + \eps}$. Since nets are nested,  $W\subseteq N_h \cap \ball(p, (4\gamma + 1)r)$. By the packing lemma (\Cref{lm:packing}), $|N_h \cap \ball(p, (4\gamma + 1)r)| = 2^{O(d)}\left(\frac{(4\gamma + 1)r}{r/(1+\eps)}\right)^d = \eps^{-O(d)}$. By \Cref{lm:reroutedInBallLargeEdges}, for each point $w \in W$, there are at most $\eps^{-O(d)}(\lceil \log(4\gamma +1) \rceil+ 1) = \eps^{-O(d)}$ edges in $\direct{M}_2$. It follows that  $|\direct{M}_2| = |W|\cdot \eps^{-O(d)} = \eps^{-O(d)}$. The lemma now follows from the fact that both $\direct{M}_1$ and $\direct{M}_2$ have $\eps^{-O(d)}$ edges. 
	\end{proof}

	\begin{claim}\label{clm:E2out}
			$|\direct{E}^{2}_\text{out}| = \eps^{-O(d)}$.
	\end{claim}
	\begin{proof}
		Let $U$ be the set of endpoints in $\ball(p,r)$ of edges in $\direct{E}^{2}_\text{out}$. By \Cref{lm:boundedDeg}, $|\direct{E}^{2}_\text{out}| = O(|U|\eps^{-O(d)})$. Let $u$ be an arbitrary point in $U$. Let $x, y$ be two points that there exists an edge $(u\rightarrow x) \in E(\direct{G_1})$ that is rerouted to $y$. That is, $(u\rightarrow y)\in \direct{E}^{2}_\text{out}$. Since $\direct{E}^{2}_\text{out}\subseteq \direct{E}^{2}_\text{long}$,  $\nd_X(u, y) \geq \gamma r$. By \Cref{clm:reroutedLength}, we have $\nd_X(u, x) \geq \frac{\nd_X(u, y)}{1 + \eps} \geq \frac{\gamma r}{1 + \eps}$. Hence, $u \in N_i$ where $i$ is the minimum index such that $r_i \geq \frac{r}{1 + \eps}$. It follow that $u \in  N_i \cap \ball(p, r)$ and hence, $U\subseteq  N_i \cap \ball(p, r)$. By the packing lemma (\Cref{lm:longEdgeEndpts}),  $|N_i \cap \ball(p, r)| = 2^{O(d)}\left(\frac{r}{r/(1+\eps)}\right)^d  = 2^{O(d)}$. We conclude that the total number of edges in $\direct{M}_2$ is bounded by $|U|\eps^{-O(d)} \leq |N_i \cap \ball(p, r)|\eps^{-O(d)} = \eps^{-O(d)}$ as claimed.
	\end{proof}
	
	Since $\direct{E}^2_\text{long} = \direct{E}^{2}_\text{in} \cup \direct{E}^{2}_\text{out}$, by \Cref{clm:E2in} and \Cref{clm:E2out}, $|E^2_\text{long}| = \eps^{-O(d)}$. As we showed \hyperlink{E1Long}{above},  $|E^1_\text{long}| = \eps^{-O(d)}$. Thus, $|E_\text{long}| = |E^1_\text{long}| + |E^2_\text{long}| = \eps^{-O(d)}$. Since $|E_{\text{short}}| = \eps^{-O(d)}$, the number of edges of length at least $r$ cut by $\ball(p,r)$ is $|E_{\text{short}}|+ |E_{\text{long}}| = \eps^{-O(d)}$, as desired. 
\end{proof}

\section{Separators of Greedy Spanners for Doubling Metrics}\label{sec:doublinggreedy}

In this section, we prove \Cref{thm:strong_doubling_separator}. We begin with the following lemma which bounds the number of edges in the greedy spanner between two well separated pair.

\begin{lemma}
	\label{lm:WSPDSpEdges}
	Let $(A, B)$ be a $(4/\eps)$-separated pair in $(X,\nd_X)$. There is at most one edge in $G$ between $A$ and $B$. 
\end{lemma}
\begin{proof}
	Suppose that there are two edges $(u, v)$ and $(u', v')$ such that $u$ and $u'$ are in $A$ and $v$ and $v'$ are in $B$. W.l.o.g, we assume that $\nd_X(u, v) \geq \nd_X(u', v')$. By the definition of separated pairs, the diameters of $A$ and $B$ are at most $(\eps/4)\nd_X(A,B)$. Thus, $\max\{\nd_X(u,u'),\nd_X(v,v')\} \leq (\eps/4)\nd_X(A,B)$. By \Cref{fact:path-greedy}, $\nd_{G-(u,v)}(u,u')\leq (1+\eps)\nd_X(u,u') \leq (1+\eps)(\eps/4)\nd_X(A,B)$. Similarly,  $\nd_{G-(u,v)}(v,v')\leq (1+\eps)\nd_X(v,v') \leq (1+\eps)(\eps/4)\nd_X(A,B)$.  We have:
	\begin{equation}
		\begin{split}
			\nd_{G  - (u, v)}(u, v) &\leq \underbrace{\nd_{G - (u, v)}(u, u')}_{\leq~ (1+\eps) \eps \nd_X(A, B)/4} + \nd_{G - (u, v)}(u', v') + \underbrace{\nd_{G - (u, v)}(v', v)}_{\leq~ (1+\eps) \eps \nd_X(A, B)/4}\\
			&\leq 2(1 + \eps)\eps \nd_X(A, B)/4 + \nd_X(u', v') \leq \eps \nd_X(A, B) + \nd_X(u', v') \quad\mbox{(since $\eps < 1$)}\\
			&< (1 + \eps)\nd_X(u, v) \quad \mbox{(since $\nd_X(A, B) \leq \nd_X(u', v') \leq \nd_X(u, v)$)}~.
		\end{split}
	\end{equation}
	Thus, $(1 + \eps)\nd_X(u, v) > \nd_{G - (u, v)}(u, v)$, contradicting that $G$ is a $(1+\eps)$-spanner (\Cref{fact:path-greedy}). 
\end{proof}

We obtain the following corollary of \Cref{lm:WSPDSpEdges}.

\begin{corollary}\label{cor:WSPDSEdges}Let $(A, B)$ be a $\beta$-separated pair in $(X,\nd_X)$. Then, there are at most $ O\left(\frac{1}{\eps \beta}\right)^{2d}$ edges in $G$ between $A$ and $B$.
\end{corollary}
\begin{proof}Let $r_{AB} = \max\{\diam(A),\diam(B)\}$. By the definition of $\beta$-separated pairs, $r_{AB} \leq \frac{\nd_X(A, B)}{\beta}$. Let $N_A$ ($N_B$) be a $\left(\frac{\eps\beta r_{AB}}{8}\right)$-net of $A$ ($B$).  By \Cref{lm:packing}, it holds that:
	\begin{equation} \label{eq:NAsize}
		|N_A| \leq 2^{O(d)}\left(\frac{\diam(A)}{\eps \beta r_{AB}/8}\right)^{d} \leq 2^{O(d)}\left(\frac{r_{AB}}{\eps \beta r_{AB}/8}\right)^{d} = O\left(\frac{1}{\eps \beta}\right)^{d}~.
	\end{equation}
By the same argument $|N_B| = O\left(\frac{1}{\eps \beta}\right)^{d}$. Let $\mathcal{A}$ ($\mathcal{B}$) be a collection of balls with radius $\frac{\eps\beta r_{AB}}{8}$ centered at  points in $A$ ($B$). Note that balls in $\mathcal{A}$ and $\mathcal{B}$ cover $A$ and $B$, respectively.  For any two pairs of balls $(A',B') \in \mathcal{A}\times \mathcal{B}$, we observe that $(A'\cap A, B'\cap B)$ is a $(\frac{4}{\eps})$-separated pair. Thus, by \Cref{lm:WSPDSpEdges}, there is at most one edge in $G$ between $A'\cap A$ and $B'\cap B$. Therefore, the number of edges in $G$ between $A$ and $B$ is at most:
\begin{equation*}
	|\mathcal{A}|\cdot|\mathcal{B}| = |N_A|\cdot |N_B| =  O\left(\frac{1}{\eps \beta}\right)^{2d},
\end{equation*}
 by \Cref{eq:NAsize}; the corollary follows. 
\end{proof}

Next, we show that $G$ is lanky.

\begin{lemma}\label{lm:DoublingGreedyLanky}
	Let $P$ be a set of $n$ points in a metric $(X,\nd)$ of dimension $d$ and spread $\Delta$. Let $G$ be a greedy $(1+\epsilon)$-spanner for $P$ with $\eps \in (0,1/2]$. Then $G$ is $O(\eps^{-d}\log(\Delta))$-lanky.  
\end{lemma}
\begin{proof} Let $\ball(p, r)$ be a ball of radius $r$ centered at a point $p \in P$.We assume that the minimum distance is at least $1$ and that $r\geq 1/2$; otherwise, we could increase $r$ to $1/2$ without changing the number of edges that cut $\ball(p, r)$. 
	
	Let $E_\text{cut}$ be the set of edges with length at least $1$ and cutting $\ball(p,r)$. 
	We partition the set $E_\text{cut}$ into $O(\log{\Delta})$ subsets: $E_\text{cut}^i = \{(u, v) \in E_{\text{cut}} ~|~ 2^{i}r < \nd_X(u, v) \leq 2^{i+1}\}$ for $1\leq i \leq \lceil \log{\Delta} \rceil$. Let $(u, v)$ be an edge in $E_\text{cut}^i$. Since $(u, v)$ is cut by $\ball(p, r)$, at least one of the two endpoints, say $u$, of $(u,v)$ is in $\ball(p,r)$.  By the triangle inequality, it follows that	$\nd_X(p, v) \leq \nd_X(u, v) + \nd_X(p, u) \leq (2^{i+1}+1)r$.

	Let $t = \max\{\eps 2^ir,r\}/4$. Let $\mathcal{B}_{R}$ for $R > 0$ be a minimal collection of balls of radius $t$ that covered the ball $\ball(p,R)$ of radius $R$ centered at $p$;  $\mathcal{B}_{R}$ can be constructed by taking balls of radius $R$ centered at a $R$-net of  $\ball(p,R)$. By \Cref{lm:packing}, when $R = r$, we have:
	\begin{equation*}
		|\mathcal{B}_{r}| = 2^{O(d)}\left(\frac{r}{r/4}\right)^d = 2^{O(d)}
	\end{equation*}
	and when $R = (2^{i+1}+1)r$, we have:
	\begin{equation*}
	|\mathcal{B}_{(2^{i+1}+1)r}| = 2^{O(d)}\left(\frac{(2^{i+1}+1)r}{\eps2^i r/4}\right)^d = O(\eps)^{-d}
\end{equation*}
	since $\eps\leq 1/2$.
	
	Let $A$ and $B$ be two balls in $\mathcal{B}_r$ and $\mathcal{B}_{(2^{i+1}+1)r}$, respectively,  that contain two endpoints of an edge $(u,v)\in E^i_{cut}$. Observe by the triangle inequality that:
	\begin{equation*}
		\nd_X(A,B)\geq w(u,v)- \diam(A) - \diam(B) \geq w(u,v) - 4(\eps 2^i r/4) = 2^i(1-\eps)r \geq 2^{i-1}r
	\end{equation*}
	since $\eps \leq 1/2$. By \Cref{cor:WSPDSEdges} with $\beta = \frac{2^{i-1}r}{\eps 2^ir/4} = \frac{2}{\eps}$, it follows that there are $O(1)^{-d}$ edges between $A$ and $B$. Thus, we have:
	\begin{equation*}
		|E^i_{\text{cut}}| = O(1)^{-d}|\mathcal{B}_{r}|\cdot |\mathcal{B}_{(2^{i+1}+1)r}| = O(\eps)^{-d}~.
	\end{equation*}
	It follows that $|E_{\text{cut}}| = O(\eps)^{-d}\log(\Delta)$, which implies the lemma. 	
\end{proof}

We are now ready to show \Cref{thm:strong_doubling_separator}.

\StrongDoubilngSeparator*
\begin{proof}
	By \Cref{lm:DoublingGreedyLanky}, $G$ is $\tau$-lanky with $\tau = O(\eps^{-d}\log(\Delta))$, and hence $G$ is weakly $\tau$-lanky with the same parameter. By applying \Cref{cor:WSPDSEdges} with $\beta = 1$, $G$ is $\kappa$-thin with $\kappa = \eps^{-2d}$. Since a doubling metric of dimension  $d$ is $(\eta, d)$-packable with $\eta = 2^{O(d)}$ by \Cref{lm:packing}, by \Cref{cor:anotherSubGSeparator}, $H$ has a separator $S$ of size $O(2^{O(d)}\eps^{-d} (\eps^{1-d} k^{1-1/d} + \log(\Delta))) = O(k^{1-1/d})$ for fixed constants $d$ and $\eps \in (0,1/2]$. Also by \Cref{cor:subgraphSepLanky}, $S$ can be found in $O(2^{O(d)}\eps^{1-2d}k) = O(k)$ expected time.
\end{proof}

\section*{Acknowledgments} The first author thanks Shay Solomon for helpful discussions. This material is based upon work supported by the National Science Foundation under Grant No. CCF-2121952.  We thank an anonymous reviewer of  SODA22 for suggesting  derandomizing \Cref{thm:strong_euclidean_separator} and pointing us to \cite{EMT95}, which leads to  results in \Cref{subsec:derandomize}. Thanks to Joachim Gudmundsson for pointing out an erroneous claim, now removed, on linear time SSSP for $\tau$-lanky graphs in the algorithmic application section in the previous version of our paper.

	\bibliographystyle{alphaurlinit}
	\bibliography{RamseyTreewidthBib,RPTALGbib}
\appendix
\end{document}